\pdfoutput=1
\documentclass[12pt,a4paper]{amsart}
\usepackage[a4paper,inner=2.5cm,outer=2.5cm,top=2.5cm,bottom=2.5cm]{geometry}
\usepackage{amsmath,amssymb,amsthm,enumerate,mathtools,stmaryrd,enumitem
}
\usepackage{hyperref}
\hypersetup{colorlinks=true,linkcolor=blue,citecolor=teal,filecolor=magenta,urlcolor=cyan}
\usepackage[T1]{fontenc}
\usepackage[utf8]{inputenc}
\usepackage{makecell}
\usepackage{graphicx}

\usepackage{float}

\usepackage{extarrows}

\usepackage{xcolor}

\newcommand{\C}{\mathbb{C}}

\usepackage{url}

\usepackage[backend=bibtex,style=alphabetic,sorting=nyt,isbn=false,url=false,doi=true,maxalphanames=10,minalphanames=4,mincitenames=4,maxcitenames=10,minnames=4,maxnames=10,giveninits=true,maxbibnames=99]{biblatex}
\theoremstyle{plain}
\newtheorem{theorem}{Theorem}[section]
\newtheorem{proposition}[theorem]{Proposition}
\newtheorem{corollary}[theorem]{Corollary}
\newtheorem{lemma}[theorem]{Lemma}
\newtheorem{principle}[theorem]{Principle}

\theoremstyle{definition}
\newtheorem{definition}[theorem]{Definition}
\newtheorem{example}[theorem]{Example}

\makeatletter
\@addtoreset{proofpart}{theorem}
\makeatother
\theoremstyle{remark}
\newtheorem{remark}[theorem]{Remark}

\newcommand{\bW}{\mathbb{W}}
\newcommand{\bK}{\mathbb{K}}

\newcommand{\cS}{\mathcal{S}}

\newcommand{\ta}{\theta}

\newcommand{\res}{\mathop{\rm res}}

\DeclareFontFamily{U}{rcjhbltx}{}
\DeclareFontShape{U}{rcjhbltx}{m}{n}{<->rcjhbltx}{}
\DeclareSymbolFont{hebrewletters}{U}{rcjhbltx}{m}{n}

\DeclareMathSymbol{\shin}{\mathord}{hebrewletters}{152}

\newcommand{\restr}[2]{\mathop{\big\lfloor_{{#1}\to {#2}}}}
\newcommand{\set}[1]{\llbracket {#1} \rrbracket}

\title{Symplectic duality via log topological recursion}

\author[A.~Alexandrov]{A.~Alexandrov}
\address{A.~A.: Center for Geometry and Physics, Institute for Basic Science (IBS), Pohang 37673, Korea
}
\email{alex@ibs.re.kr}

\author[B.~Bychkov]{B.~Bychkov}
\address{B.~B.: Department of Mathematics, University of Haifa, Mount Carmel, 3498838, Haifa, Israel}
\email{bbychkov@hse.ru}

\author[P.~Dunin-Barkowski]{P.~Dunin-Barkowski}
\address{P.~D.-B.: Faculty of Mathematics, National Research University Higher School of Economics, Usacheva 6, 119048 Moscow, Russia; HSE--Skoltech International Laboratory of Representation Theory and Mathematical Physics, Skoltech, Bolshoy Boulevard 30 bld. 1, 121205 Moscow, Russia; and NRC “Kurchatov Institute” -- ITEP, 117218 Moscow, Russia}
\email{ptdunin@hse.ru}

\author[M.~Kazarian]{M.~Kazarian}
\address{M.~K.: Faculty of Mathematics, National Research University Higher School of Economics, Usacheva 6, 119048 Moscow, Russia; and Igor Krichever Center for Advanced Studies, Skoltech, Bolshoy Boulevard 30 bld. 1, 121205 Moscow, Russia}
\email{kazarian@mccme.ru}

\author[S.~Shadrin]{S.~Shadrin}
\address{S.~S.: Korteweg-de Vries Institute for Mathematics, University of Amsterdam, Postbus 94248, 1090GE Amsterdam, The Netherlands}
\email{S.Shadrin@uva.nl}	

\begin{document}
	
\begin{abstract}
We review the notion of symplectic duality earlier introduced in the context of topological recursion. We show that the transformation of symplectic duality can be expressed as a composition of $x-y$ dualities in a broader context of log topological recursion. As a corollary, we establish nice properties of symplectic duality: various convenient explicit formulas, invertibility, group property, compatibility with topological recursion and KP integrability. As an application of these properties, we get a new and uniform proof of topological recursion for large families of weighted double Hurwitz numbers; this encompasses and significantly extends all previously known results on this matter.
\end{abstract}
	
\maketitle
	
\setcounter{tocdepth}{3}
\tableofcontents

\section{Introduction}

Topological recursion (TR) is a procedure that associates to a so-called spectral curve data that consists of a Riemann surface $\Sigma$ and a pair of meromorphic functions $(x,y)$ a system of symmetric meromorphic differentials (called the \emph{correlation differentials} or sometimes, slightly abusing the terminology, the \emph{$n$-point functions})~\cite{EO}. It was originally developed as a computational tool for the cumulants of matrix models, but since then it evolved to a new and active area of research by itself that provides a universal interface that relates a huge class of algebraic and combinatorial enumerative problems to the intersection theory of moduli spaces and integrable systems.

In the present paper (for large families of spectral curves) we prove a closed formula which represents the correlation differentials for the spectral curve with $x$ shifted by a function of $y$ in terms of the correlation differentials for the original spectral curve. We show that this transformation preserves the KP integrability property, and prove certain other properties of such transformations. We apply our main result to uniformly prove topological recursion (or \emph{logarithmic} topological recursion, see below) for large families of weighted double Hurwitz numbers which encompass almost all (single or double) Hurwitz-type numbers discussed in the literature (see Section~\ref{sec:HurwSummary} for more details). We also provide several further applications, see below.

\subsection{Perspective and results} 

In the last few years the structural development of the theory of topological recursion turned it into a powerful computational tool that not just extends the variety of problems governed by TR but also provides explicit closed formulas for the $n$-point generating functions (as opposed to recursive computation that TR is meant to give) and closed algebraic relations among them. Among the most important developments are the following ones:
\begin{description}
	\item[$x-y$ duality] The $x-y$ duality traces what happens with the correlation differentials under the swap of $x$ and $y$. An explicit closed formula was conjectured in~\cite{borot2023functional}, and it was soon proved in~\cite{hock2022xy} in a special case, substantially simplified in~\cite{hock2022simple}, and proved in general case in~\cite{ABDKS1}. This formula gained further reformulations and applications (in particular, to KP integrability) in~\cite{hock2023laplace,ABDKS3}.
 	\item[symplectic duality] Symplectic duality was introduced in~\cite{BDKS4} as a generalization of the $x-y$ duality that depends on an additional rational function. It was used mostly in the context of vacuum expectation values to prove topological recursion for weighted double Hurwitz numbers / weighted maps and constellations with internal faces (see also~\cite{bonzom2022topological}, though the latter reference does not cover the most general cases like the $r$-spin Hurwitz numbers), and for the so-called generalized fully simple maps in~\cite{ABDKS2}.
	\item[LogTR] It was first observed in~\cite{hock2023xy} and further structurally developed in~\cite{ABDKS4} that one can extend the realm of applications of $x-y$ duality in topological recursion beyond the case of meromorphic $x$ and $y$. It is sufficient to assume that $dx$ and $dy$ are meromorphic, thus $x$ and $y$ are allowed to have logarithmic singularities. In this case the formulation of topological recursion needs an adjustment which we call the log topological recursion (LogTR). 
\end{description}

The main goal of this paper is to show that more general symplectic duality can, in fact, be represented as a sequence of $x-y$ dualities, with eventual corrections coming from the necessity to use the log topological recursion. This unifies the existing approaches to explicit closed formulas for the $n$-point differentials, identifies competing closed formulas provided in the literature (which is especially remarkable since these formulas naturally occur in quite different contexts), upgrades symplectic duality to a group action on the spectral curve data and associated differentials (and brings it far beyond the original setup of vacuum expectation values where it was first developed), and shows that the symplectic duality preserves the KP integrability. The results of the present paper also serve as another argument in favor of the claim that LogTR as defined in~\cite{ABDKS4} is indeed \emph{the} natural extension of TR to spectral curves with logarithmic singularities. 

%As an application we reprove the topological recursion for various families of weighted Hurwitz numbers (and we go beyond the known cases as we can now remove various technical restrictions imposed before),
As an application of our results on symplectic duality we get a new and uniform proof of (log) topological recursion for large families of weighted double Hurwitz numbers. This encompasses all previously known results on this matter (including the families studied in \cite{BDKS2}; also see \cite[Section~5.2]{BDKS2} in particular for the special cases of those families studied in the prior literature), but also quite significantly extends beyond that, as we can, in particular, now remove various technical restrictions imposed before. This proof is also quite a bit nicer than everything available before, since even in \cite{BDKS2} one had to carefully analyze various possible poles (separately for different families) to prove the projection property, while now we give a much simpler and more general  and more conceptual proof.

We also give a new way how to identify various formulas for Kontsevich--Witten $n$-point functions, and we give a new proof of a recent theorem of Bouchard, Kramer, and Weller~\cite{BKW} that transalgebraic topological recursion reproduces in the limit the so-called Atlantes Hurwitz numbers (which do not satisfy the original TR by themselves), introduced in~\cite{ALS}.

\subsection{Organization of the paper} In Section~\ref{sec:TopologicalRecursion-Ord-n-Log} we remind the definitions of the ordinary and logarithmic topological recursion. In Section~\ref{sec:Review-x-y-duality} we provide a review of the main results on the $x-y$ duality. In Section~\ref{sec:SymplecticDuality} we recall the notion of symplectic duality. In Section~\ref{sec:SymplDual-via-xy} we state and prove the main result (Theorem \ref{thm:sympl-from-xy}) of this paper that reduces symplectic duality to $x-y$ duality.  In Section~\ref{sec:CompatibilityTR} we append the results of the previous section with an analysis of the (logarithmic) projection property in order to establish the connection to (log) topological recursion; as an application we get a new proof of (log) topological recursion for various general families of weighted double Hurwitz numbers. Finally, in Section~\ref{sec:MoreExamples} we collect some further examples.

\subsection{Notation} Throughout the text we use the following notation:
\begin{itemize}
	\item $\set{n}$ denotes $\{1,\dots,n\}$.
	\item $z_I$ denotes $\{z_i\}_{i\in I}$ for $I\subseteq \set{n}$.
	\item $[u^d]$ denotes the operator that extracts the corresponding coefficient from the whole expression to the right of it, that is, $[u^d]\sum_{i=-\infty}^\infty a_iu^i \coloneqq a_d$.
	\item
	$\restr{u}{v}$ denotes the operator of substitution (or restriction), that is,\\ $\restr{u}{v} f(u) \coloneqq f(v)$.
	\item $\cS(u)$ denotes $u^{-1}(e^{u/2} - e^{-u/2})$.
	\item $d_i$ stands for taking the differential with respect to the $i$-th variable.

\end{itemize}

\subsection{Acknowledgments}
A.~A. was supported by the Institute for Basic Science (IBS-R003-D1).
B.B. was supported by the ISF Grant 876/20.  B.B., P.D.-B., and M.K. were supported by the Russian Science Foundation (grant No. 24-11-00366).
S.~S. was supported by the Netherlands Organization for Scientific Research. A.A. is grateful for the hospitality to the
Skoltech, where part of this research was carried out.

We thank V.~Bouchard, A.~Hock, R.~Kramer, and Q.~Weller for useful discussions. We are also grateful to the anonymous referees for their helpful remarks and suggestions.

\section{Topological recursion and logarithmic topological recursion}
\label{sec:TopologicalRecursion-Ord-n-Log}

\subsection{Topological recursion}
Let us recall the definition of topological recursion (TR) in terms of the loop equations and projection property. We follow the exposition in~\cite[Section 5]{ABDKS1}. 

Let $x$ and $y$ be meromorphic functions on a compact Riemann surface $\Sigma$, and $\{\omega^{(g)}_{m}\}$, $g\geq 0$, $m\geq 1$, $2g-2+m>0$ be a given system of symmetric meromorphic differentials ($\omega^{(g)}_{m}$ is defined on $\Sigma^m$). We also define $\omega^{(0)}_1 \coloneqq ydx$ and $\omega^{(0)}_2\coloneqq B$, where $B$ is the so-called Bergman kernel, that is, the unique symmetric meromorphic bi-differential on $\Sigma^2$ with a double pole along the diagonal with bi-residue $1$ and no further singularities, whose $\mathfrak{A}$-periods (for some choice of $\mathfrak{A}$-cycles) are all equal to $0$.

\begin{remark}
	Note that we use the definition $\omega^{(0)}_1 \coloneqq ydx$ in the present paper, %as customary in most TR-related papers, 
	as opposed to % 
	a different convention $\omega^{(0)}_1 \coloneqq -ydx$ used, e.g., in \cite{ABDKS1,ABDKS4}. 
	
	Both conventions are widely used in the literature.
\end{remark}

Assume that all zeros of $dx$ are simple and $y$ is regular at the zeros of $dx$ and the zeros of $dx$ and $dy$ are disjoint. Denote the zeros of $dx$ by $p_1,\dots,p_N\in \Sigma$ and let $\sigma_i$ be the deck transformations of $x$ near $p_i$.

\begin{definition} We say that the system of symmetric meromorphic differentials $\{\omega^{(g)}_m\}$, $g\geq 0$, $m\geq 1$ satisfies
	\begin{itemize}
		\item the \emph{linear loop equations}, if for any $g,m\geq 0$, $i=1,\dots,N$
		\begin{equation} \label{eq:LLE-original}
			\omega^{(g)}_{m+1}(z_{\llbracket m \rrbracket},z) + \omega^{(g)}_{m+1}(z_{\llbracket m \rrbracket},\sigma_i(z))
		\end{equation}
		is holomorphic at $z\to p_i$ and has at least a simple zero in $z$ at $z=p_i$;
		\item the \emph{quadratic loop equations}, if for any $g,m\geq 0$, $i=1,\dots,N$, the quadratic differential in $z$
		\begin{equation} \label{eq:QLE-original}
			\omega^{(g-1)}_{m+2}(z_{\llbracket m \rrbracket},z,\sigma_i(z)) + \sum_{\substack{g_1+g_2=g\\ I_1\sqcup I_2 = \llbracket m \rrbracket
			}} \omega^{(g_1)}_{|I_1|+1} (z_{I_1},z)\omega^{(g_2)}_{|I_2|+1} (z_{I_2},\sigma_i(z))
		\end{equation}
		is holomorphic at $z\to p_i$ and has at least a double zero at $z=p_i$.
	\end{itemize}
\end{definition}

\begin{remark}\label{rem:QLEDelta}
	Note that if it is already known that a given system of symmetric meromorphic differentials $\{\omega^{(g)}_m\}$, $g\geq 0$, $m\geq 1$ satisfies the linear loop equations, then the quadratic loop equations are equivalent to the following statement (see~\cite{DKPS-rspin}): for any $g,m\geq 0$, $i=1,\dots,N$, the expression
	\begin{equation} \label{eq:QLE-newform}
		\restr{w}{z}\Delta_{i,z}\Delta_{i,w}\left(\omega^{(g-1)}_{m+2}(z_{\llbracket m \rrbracket},z,w) + \sum_{\substack{g_1+g_2=g\\ I_1\sqcup I_2 = \llbracket m \rrbracket
		}} \omega^{(g_1)}_{|I_1|+1} (z_{I_1},z)\omega^{(g_2)}_{|I_2|+1} (z_{I_2},w)\right)
	\end{equation}
	is holomorphic at $z\to p_i$ and has at least a double zero at $z=p_i$. Here $\Delta_{i,z} f(z) := f(z)-f(\sigma_i(z))$.
\end{remark}

\begin{definition} We say that $\{\omega^{(g)}_{m}\}$ satisfies the \emph{blobbed topological recursion} \cite{BS-blobbed}, if it satisfies the linear and quadratic loop equations.
\end{definition}

Under the assumption that all zeros $\{p_1,\dots,p_N\}$ of $dx$ are simple, the statement that $\{\omega^{(g)}_{m}\}$ satisfies the blobbed topological recursion implies that the principal parts near the zero locus of $dx$ of each $\omega^{(g)}_{m}$, $2g-2+m>0$, in each variable are uniquely determined by $\omega^{(g')}_{m'}$ with $2g'-2+m' < 2g-2+m$.

\begin{definition} \label{def:PP}
	We say that $\{\omega^{(g)}_{m}\}$ satisfy the original \emph{topological recursion} (TR) on the spectral curve $(\Sigma,x,y)$, if in addition to the blobbed topological recursion they satisfy the so-called \emph{projection property}: for any $g\geq 0$, $m\geq 1$, $2g-2+m>0$
	\begin{equation}
		\omega^{(g)}_{m}(z_{\set{m}}) = \sum_{i_1,\dots,i_m=1}^N \Bigg(\prod_{j=1}^m \res_{z_j'\to p_{i_j}} \int^{z_j'}_{p_{i_j}} B(\cdot,z_j)\Bigg) \omega^{(g)}_{m}(z'_{\set{m}}).
	\end{equation}
\end{definition}

In other words, we say that in the case of original topological recursion the meromorphic differentials $\omega^{(g)}_{m}$, $2g-2+n>0$, have no other poles than at $\{p_1,\dots,p_N\}$, and all their $\mathfrak{A}$-periods are equal to zero.

\subsection{Logarithmic topological recursion}\label{sec:LogTRDef}

\subsubsection{Motivation of log topological recursion}
Topological recursion relations depend on the differentials $dx$, $dy$ only, but not the functions $x$ and $y$ themselves, and they are well defined if $dx$ and $dy$ are global meromorphic differentials with possibly nonzero residues. In the latter case the functions $x$ and $y$ are not univalued and posses logarithmic singularities. In the presence of this sort of singularities it is useful to consider, along with the original topological recursion, its variation that we call \emph{logarithmic topological recursion} (we also call it log topological recursion or just LogTR), introduced in~\cite{ABDKS4} and motivated by the results of~\cite{hock2023xy}. It turns out that LogTR (which coincides with the original TR for meromorphic $x$ and $y$ but differs if $y$ has logarithmic singularities of certain type described below) is the proper version of topological recursion if one wants to consider the $x-y$ duality or the symplectic duality which are discussed below. That is, as proved in~\cite{ABDKS4}, the $x-y$ duality formula holds for LogTR in the presence of logarithmic singularities in the same form as for meromorphic $x$ and $y$. In the present paper we prove a similar statement regarding the more general \emph{symplectic duality}. Since these statements have numerous implications for many enumerative problems, as discussed in~\cite{hock2023xy,ABDKS4} and the present paper respectively, one might argue that LogTR is \emph{the} correct definition in the presence of logarithmic singularities.

\subsubsection{Definition of log topological recursion}
Assume that $dx$ and $dy$ are globally defined meromorphic $1$-forms on a compact Riemann surface $\Sigma$ with possibly non-vanishing residues at some points. We still assume that all zeros $\{p_1,\dots,p_N\}$ of $dx$ are simple, $dy$ is regular at the zero locus of $dx$ and the zero loci of $dx$ and $dy$ are disjoint.

\begin{definition}\label{def:necTR}
	We call the above conditions on $(\Sigma,x,y)$ the \emph{necessary TR conditions}.
\end{definition}

\begin{definition}
	We say that the primitive $y$ of the differential $dy$ posses \emph{logarithmic singularity} at some point~$a$ on~$\Sigma$ if $dy$ has a pole at~$a$ with nonzero residue. A logarithmic singularity of~$y$ is called \emph{LogTR-vital} if this pole of $dy$ is simple and $dx$ has no pole at this point.
\end{definition}

Let $a_1,\dots,a_M$ be the LogTR-vital singular points of $y$. We denote the residues of~$dy$ at these points by $\alpha_1^{-1},\dots,\alpha_M^{-1}$, respectively. That is, the principal part of $dy$ near $a_i$ is given by $\alpha_i^{-1} dz / (z-a_i)$ in any local coordinate~$z$.% such that $dz(a_i)\not=0$.

Consider the germs of the principal parts of meromorphic $1$-forms in the neighborhoods of $a_i$, $i=1,\dots,M,$ defined as the principal parts of coefficients of positive powers of $\hbar$ in the following expressions:
\begin{align}\label{eq:PrincipalPartsLog}
	\left( \frac{1}{\alpha_i\cS(\alpha_i\hbar \partial_{x})}\log(z-a_i)\right)dx.
\end{align}

\begin{definition}\label{def:log-proj} Let $\{{}^{\log{}}\omega^{(g)}_{m}\}$, $g\geq 0$, $m\geq 1$, $2g-2+m>0$ be a given system of symmetric meromorphic differentials, and set ${}^{\log{}}\omega^{(0)}_1 = ydx$, ${}^{\log{}}\omega^{(0)}_2 = B$. We say that $\{{}^{\log{}}\omega^{(g)}_{m}\}$ satisfies the \emph{logarithmic topological recursion} (LogTR) on the spectral curve $(\Sigma,x,y)$, if  $\{{}^{\log{}}\omega^{(g)}_{m}\}$ satisfy the blobbed topological recursion and additionally they satisfy the so-called \emph{logarithmic projection property}: for any $g\geq 0$, $m\geq 1$, $2g-2+m>0$
	\begin{align} \label{eq:LogProjection}
		{}^{\log{}}\omega^{(g)}_{m}(z_{\set{m}}) & = \sum_{i_1,\dots,i_m=1}^N \Bigg(\prod_{j=1}^m \res_{z_j'= p_{i_j}} \int^{z_j'}_{p_{i_j}} B(\cdot,z_j)\Bigg) {}^{\log{}}\omega^{(g)}_{m}(z'_{\set{m}})
		\\ \notag & \quad
		+ \delta_{m,1}[\hbar^{2g}] \sum_{i=1}^M  \res_{z'= a_i} \Bigg(\int^{z'}_{a_i} B(\cdot,z_1)\Bigg)
		\left(\frac{1}{\alpha_i\cS(\alpha_i\hbar \partial_{x'})}\log(z'-a_i)\right)dx'.
	\end{align}
\end{definition}

In other words, we say that in the case of logarithmic topological recursion the meromorphic differentials ${}^{\log{}}\omega^{(g)}_{m}$, $2g-2+m>0$, $m\geq 2$, have no other poles than at $\{p_1,\dots,p_N\}$  and all their $\mathfrak{A}$-periods are equal to zero.
For $g\geq 1$ and $m=1$ we have a bit different setup: they
have poles at $\{p_1,\dots,p_N\}$, and also at the points $\{a_1,\dots,a_M\}$. The principal parts at the latter points are equal to the principal parts of expressions given by~\eqref{eq:PrincipalPartsLog}, and all their $\mathfrak{A}$-periods are equal to zero.

\begin{remark} Note that if $a$ is a simple pole of~$dy$ and $dx$ has also a pole at this point, that is, singularity of $y$ is not LogTR-vital, then %the coefficient of $\hbar^{2g}$ of
	the principal part of~$[\hbar^{2g}]\bigl( \frac{1}{\alpha\cS(\alpha\hbar \partial_{x})}\log(z-a)\bigr)dx$ with $g>0$ is equal to zero,
	%for $g\geq 1$,
	so we could formally include these points to the list of LogTR-vital singularities of~$y$ without changing relation of recursion. But it is important to underline that such singular points do not contribute to the recursion and do not affect Equation~\eqref{eq:LogProjection}.
\end{remark}

\begin{remark}
	The poles of order greater than $1$ are excluded from the logarithmic projection property by similar local considerations. Indeed, assume that the meromorphic form~$dy$ varies in a family such that two simple poles collapse together producing a pole of order~$2$. Then, it is easy to see by local computations that the residues at these two poles tend to infinity when the points approach one another. It follows that the limiting contribution of these two points to~\eqref{eq:PrincipalPartsLog} is trivial, independently of whether the limiting pole of order two has a residue or not. The same argument also works for the higher order poles.
\end{remark}

\begin{remark}
	In the remaining part of the paper we always use LogTR (not the original TR), but we omit ``$\log$''  in the notation for the $n$-point functions, i.e. we write $\omega^{(g)}_{m}$ in place of ${}^{\log{}}\omega^{(g)}_{m}$.
\end{remark}

\section{Review of \texorpdfstring{$x-y$}{x-y} duality} \label{sec:Review-x-y-duality}

In this section we remind the reader a definition and the main properties of the $x-y$ duality, see \cite{ABDKS1,ABDKS3,ABDKS4} for details. 

\subsection{Spectral curve data}

A \emph{spectral curve data} is a compact Riemann surface (a smooth complex algebraic curve) $\Sigma$ equipped with a pair of functions $(x,y)$ on it such that $dx$ and $dy$ are meromorphic $1$-forms. We do not assume that $x$ and $y$ are themselves meromorphic; these functions could be multivalued and possess logarithmic singularities.

Let $\{\omega^{(g)}_n\}$ be a system of symmetric meromorphic $n$-differentials on~$\Sigma^n$ defined for $g\ge0$, $n\ge1$, $(g,n)\ne(0,1)$ such that $\omega^{(g)}_n-\delta_{(g,n),(0,2)}\frac{dx_1dx_2}{(x_1-x_2)^2}$ is regular on the diagonals in $\Sigma^n$. We extend this system of differentials to the case $(g,n)=(0,1)$ by setting
\begin{equation}
\omega^{(0)}_1=y\,dx
\end{equation}
even though this form may not be well defined as a meromorphic one. Let us stress here that we do not assume that the differentials satisfy any topological recursion.

\begin{remark}
The definition of the spectral curve data that can be found in the literature often includes also a choice for the 2-differential $B=\omega^{(0)}_2$. In most cases of our particular interest, we have $\Sigma=\C \mathrm{P}^1$ and $\omega^{(0)}_2$ is uniquely defined by $\omega^{(0)}_2=\frac{dz_1dz_2}{(z_1-z_2)^2}$ for an arbitrary affine coordinate~$z$. By that reason we will often denote by $z_i$ a point on the $i$th copy of the Cartesian product $\Sigma^n$ identifying this point with the value of the coordinate $z$ at this point, and denote $x_i=x(z_i)$, $y_i=y(z_i)$. Note, however, that the basic definitions and relations of this paper can be applied to a spectral curve of any genus (or to a different choice for $\omega^{(0)}_2$ in the case of a rational spectral curve).
\end{remark}

\subsection{Extended \texorpdfstring{$n$}{n}-point differentials}

The $x-y$ duality construction involves summation over hypergraphs as an intermediate step.

\begin{definition}\label{def:OmegabW}
Given a system of differentials $\{\omega^{(g)}_n\}$, the \emph{extended $n$-point differentials} $\Omega_n=\sum_{g=0}^\infty\hbar^{2g-2+n}\Omega^{(g)}_n$, $\bW_n=\sum_{g=0}^\infty\hbar^{2g-2+n}\bW^{(g)}_n$ are defined by
\begin{align}\label{eq:Omegan}	
\Omega_{n}(z^+_{\set{n}},z^-_{\set{n}})
	 &=\prod_{i=1}^n\tfrac{\sqrt{dx_i^+ dx_i^-}}{x_i^+-x_i^-}\;
\sum_{\Gamma} \tfrac{1}{|\mathrm{Aut}(\Gamma)|}
	\prod_{e\in E(\Gamma)} \sum_{\tilde g=0}^\infty \hbar^{2\tilde g-2+|e|}
\int\limits_{z_{e(1)}^-}^{z_{e(1)}^+}\dots \int\limits_{z_{e(|e|)}^-}^{z_{e(|e|)}^+}
\tilde\omega^{(\tilde g)}_{|e|},
\\\label{eq:bWn}	
\bW_{n}(z_{\set{n}},u_{\set{n}})
	 &=\prod_{i=1}^n  \Bigl(e^{u_i(\cS(u_i\hbar\partial_{x_i})-1)y_i}
\restr{z_i^{\pm}}{e^{\pm \frac{u_i\hbar}{2}\partial_{x_i}}z_i}\Bigr)
\;\Omega_{n}(z^+_{\set{n}},z^-_{\set{n}}).
\end{align}
Here 
\begin{itemize}
\item $x_i=x(z_i)$, $y_i=y(z_i)$, where $z_i$ is a point on the $i$-th copy of $\Sigma$ in $\Sigma^n$, and similarly $x_i^\pm=x(z_i^\pm)$, where $(z_i^+,z_i^-)$ are points near the diagonal on the $i$-th copy of $\Sigma^2$ in $(\Sigma^2)^n$. 

\item The sum is taken over all connected  bipartite graphs (or hypergraphs or just graphs with multiedges, for brevity)  $\Gamma$ with $n$ labeled vertices and multiedges of arbitrary index $\geq 1$ (in the bipartite graph picture we refer to the first set of vertices as to just ``vertices'', and to the second set as ``multiedges''), where the index of a multiedge is the number of its ``legs'' (i.e. the valence of the respective vertex in the bipartite graph picture) and we denote it by $|e|$. Multiple multiedges connecting the same set of vertices are allowed, and multiple ``legs'' of a multiedge going to the same vertex is also allowed.

\item For a multiedge $e$ with index $|e|$  we control its attachment to the vertices by the associated map $e\colon \set{|e|} \to \set{n}$ that we denote also by $e$, abusing notation (so $e(j)$ is the label of the vertex to which the $j$-th ``leg'' of the multiedge $e$ is attached).

\item Terms with $(g,n)=(0,1)$ are excluded from the summation, that is, we set $\tilde\omega^{(0)}_1=0$, and the meaning of $\tilde \omega^{(0)}_{|e|}$ in the case $|e|=2$ depends on whether the two legs $e(1)$ and $e(2)$ are the same or not. 

\item We set $\tilde\omega^{(0)}_{|e|}(z_1,z_2)=\omega^{(0)}_2(z_1,z_2)- \frac{dx_1dx_2}{(x_1-x_2)^2}$ if $e(1)=e(2)$, and we set  $\tilde\omega^{(0)}_{|e|}=\omega^{(0)}_2$ otherwise. For all $(g,n)\neq (0,2)$ or $(0,1)$ we simply have $\tilde \omega^{(g)}_{n} =  \omega^{(g)}_{n}$.

\item $|\mathrm{Aut}(\Gamma)|$ stands for the number of automorphisms of $\Gamma$.
\end{itemize}
\end{definition}

\begin{remark}\label{rem:bW-polynomiality}
The arguments $(z^+_{i},z^-_{i})$ of $\Omega_{n}$ are treated as a point of the Cartesian square $\Sigma^2$ and are considered in a certain vicinity (tubular neighborhood) of the diagonal $z_i^+=z_i^-$. The arguments $(z_i,u_i)$ of $\bW_n$ provide a particular parametrization of this vicinity through the substitutions
\begin{equation}
z_i^{+}=e^{\frac{u_i\hbar}{2}\partial_{x_i}}z_i,\qquad
z_i^{-}=e^{-\frac{u_i\hbar}{2}\partial_{x_i}}z_i.
\end{equation}
\end{remark}

\begin{remark}
The half-differentials $\sqrt{dx_i^\pm}$ entering $\Omega_n$ are understood formally. They attain a more definite meaning after the above substitution only. Namely, we have
\begin{equation}
\restr{z_i^{\pm}}{e^{\pm \frac{u_i\hbar}{2}\partial_{x_i}}z_i} x_i^\pm=x_i\pm\frac{u_i\hbar}{2},
\qquad\restr{z_i^{\pm}}{e^{\pm \frac{u_i\hbar}{2}\partial_{x_i}}z_i}
\frac{\sqrt{dx_i^+dx_i^-}}{x_i^+-x_i^-}=\frac{dx_i}{u_i\hbar}.
\end{equation}
\end{remark}

\begin{remark}\label{rmk3.5} We regard $\bW_n$'s as $n$-differentials on $\Sigma^n$ depending on additional parameters $u_1,\dots,u_n$.
One can see that they are expressed in terms of the differentials $\omega^{(\tilde g)}_{\tilde n}$ and their derivatives (no actual integration is needed). This fact follows from the equality
\begin{equation}
\prod_{i=1}^{\tilde n}\restr{\tilde z_i^{\pm}}{e^{\pm \frac{u_{e(i)}\hbar}{2}\partial_{x_{e(i)}}}z_{e(i)}}
\;\int\limits_{\tilde z_{1}^-}^{\tilde z_{1}^+}\dots \int\limits_{\tilde z_{\tilde n}^-}^{\tilde z_{\tilde n}^+}
\tilde\omega^{(\tilde g)}_{\tilde n}
=\prod_{i=1}^{\tilde n}\bigl(\restr{\tilde z_i}{z_{e(i)}} u_{e(i)}\hbar\,\cS(u_{e(i)}\hbar\,\partial_{\tilde x_i})\bigr)\;
\frac{\tilde\omega^{(\tilde g)}_{\tilde n}(\tilde z_{\set{\tilde n}})}
{\prod_{i=1}^{\tilde n} d\tilde x_i}.
\end{equation}
Note also that the coefficient of each power of~$\hbar$ in $\bW_n-\delta_{n,1}\frac{dx_1}{u_1\hbar}$ is a polynomial in $u$-variables, thus $\bW^{(g)}_n$ for any $g$ and $n$ is a differential polynomial in $\omega$'s.
\end{remark}

\begin{remark}
	A nice property of the (half-)differentials~$\Omega_n$ is that they are actually independent of $x$ and $y$ and are uniquely determined by the differentials $\omega^{(g)}_n$ with $(g,n)\ne(0,1)$. Indeed, a different choice of the $x$-function changes both the factors $\frac{\sqrt{dx_i^+dx_i^-}}{x_i^+-x_i^-}$ and the regularization terms $\frac{d\tilde x_1 d\tilde x_2}{(\tilde x_1-\tilde x_2)^2}$ entering the definition of $\tilde\omega^{(0)}_2$. An easy computation shows that these two changes compensate one another, see~\cite[Equation (3.10)]{ABDKS3}.
\end{remark}

\begin{example}\label{ex:bWtrivial} Assume that $\omega^{(g)}_n=0$ for all $(g,n)$ except $\omega^{(0)}_1=y dx$ and $\omega^{(0)}_2=\frac{dz_1dz_2}{(z_1-z_2)^2}$. Then, using the Cauchy determinant formula (see e.~g.~\cite{ABDKS4,hock2023xy}), one computes
\begin{equation}
\Omega_n=\prod_{i=1}^n \sqrt{dz_i^+dz_i^-}
\;(-1)^{n-1}\!\!\sum_{\sigma\in\{n\text{-cycles}\}}\prod_{i=1}^n\frac{1}{z_i^+-z_{\sigma(i)}^-}.
\end{equation}

Next, assume that $x(z)=z$. Then
\begin{equation}\label{eq:Wtrivz}
\bW_n=\prod_{i=1}^n \Bigl(e^{u_i(\cS(u_i\hbar\partial_{z_i})-1)y_i}dz_i\Bigr)
\;(-1)^{n-1}\!\!\sum_{\sigma\in\{n\text{-cycles}\}}\prod_{i=1}^n\frac{1}{z_i+\frac{u_{i}\hbar}{2}-z_{\sigma(i)}+\frac{u_{\sigma(i)}\hbar}{2}}.
\end{equation}

In a similar situation with $x(z)=\log z$ one has
\begin{equation}\label{eq:Wtrivlogz}
\bW_n=\prod_{i=1}^n \Bigl(e^{u_i(\cS(u_i\hbar z_i\partial_{z_i})-1)y_i}dz_i\Bigr)
\;(-1)^{n-1}\!\!\sum_{\sigma\in\{n\text{-cycles}\}}\prod_{i=1}^n\frac{1}{e^{\frac{u_{i}\hbar}{2}}z_i-e^{-\frac{u_{\sigma(i)}\hbar}{2}}
z_{\sigma(i)}}.
\end{equation}
\end{example}

\begin{remark}
In this paper our sign conventions are different comparing to the previous papers including \cite{ABDKS3,ABDKS4}. Moreover, the definition of the extended $n$-point differentials $\Omega_n$ is slightly different comparing to \cite{ABDKS3}, were by $\Omega_n$ we denote corresponding functions. 
\end{remark}

\subsection{\texorpdfstring{$x-y$}{x-y} duality and its basic properties}

Starting from a given system of differentials $\{\omega^{(g)}_n\}$,  the $x-y$ duality construction produces a new one denoted by $\{\omega^{\vee,(g)}_n\}$ with the role of $x$ and $y$ functions swapped. That is, we have
\begin{align}
(x^\vee,y^\vee)&=(y,x),\\
\omega^{\vee,(0)}_1&=y^\vee\,dx^\vee=x\,dy.
\end{align}

\begin{definition} \label{def:xy-duality}
The $x-y$ dual differentials are defined by
\begin{equation}
\omega^{\vee,(g)}_n(z_{\set n})
=(-1)^n[\hbar^{2g-2+n}]\prod_{i=1}^n
\biggl(\sum_{r=0}^\infty \bigl(-d_i\tfrac{1}{dy_i}\bigr)^{r}[u_i^r]\biggr)
    \;\bW_n(z_{\set n},u_{\set n})
\end{equation}
for $(g,n)\ne(0,1)$ and  $\omega^{\vee,(0)}_1=x\, dy$, where $\bW_n$ is the $n$-point differential of the previous section and $d\frac{1}{dx}$ is the operator acting in the space of meromorphic differentials and sending a given $1$-form $\eta$ to $d\frac{\eta}{dx}$.
\end{definition}

The polynomiality property of Remark~\ref{rmk3.5} implies that the expression for $\omega^{\vee,(g)}_n$ involves a finite number of summands only. For example, for small $(g,n)$ we have
\begin{align}
\omega^{\vee,(0)}_2&=\omega^{(0)}_2,
\\\omega^{\vee,(0)}_3&=-\omega^{(0)}_3
+d_1\tfrac{\omega^{(0)}_2(z_1,z_2)\omega^{(0)}_2(z_1,z_3)}{dy_1dx_1}
+d_2\tfrac{\omega^{(0)}_2(z_2,z_3)\omega^{(0)}_2(z_2,z_1)}{dy_2dx_2}
+d_3\tfrac{\omega^{(0)}_2(z_3,z_1)\omega^{(0)}_2(z_3,z_2)}{dy_3dx_3},
\\\omega^{(0),{\rm reg}}_2&=\omega^{(0)}_2-\tfrac{dx_1dx_2}{(x_1-x_2)^2},
\\\omega^{\vee,(1)}_1&=-\omega^{(1)}_1+d\biggl(
\frac{\omega^{(0),{\rm reg}}_2(z,z)}{2\,dx\,dy}
-\frac1{24}\partial_y^2\frac{dy}{dx}
\biggr).
\end{align}

Here are the basic properties of $x-y$ duality, for more details see \cite{ABDKS1,ABDKS3,ABDKS4}.

\begin{theorem}\label{Th:x-y}
	\hspace{1em}
\begin{enumerate}[label={\arabic*.},itemsep=5pt]
	
\item %1. 
The dual differentials $\omega^{\vee,(g)}_n$ are meromorphic, symmetric, and possess the property that $\omega^{\vee,(g)}_n-\delta_{(g,n),(0,2)}\frac{dx^\vee_1dx^\vee_2}{(x^\vee_1-x^\vee_2)^2}$ are regular on the diagonals.

\item %2. 
This correspondence is a duality indeed: the original differentials $\omega^{(g)}_n$ can be recovered from the dual ones by a similar formula with the role of~$x$ and~$y$ functions swapped
\begin{equation}\label{omegaveeW}
\omega^{(g)}_n(z_{\set n})
=(-1)^n[\hbar^{2g-2+n}]\prod_{i=1}^n
\biggl(\sum_{r=0}^\infty \bigl(-d_i\tfrac{1}{dx_i}\bigr)^{r}[v_i^r]\biggr)
    \;\bW^\vee_n(z_{\set n},v_{\set n}),
\end{equation}
$(g,n)\not=(0,1)$, where the differentials $\bW^\vee_n$ are obtained from the differentials $\omega^{\vee,(\tilde g)}_{\tilde n}$ through the differentials $\Omega^\vee_n$ in a way similar to \eqref{eq:Omegan}--\eqref{eq:bWn} by summation over graphs.

\item %3. 
Moreover, the extended dual differentials  $\bW_n$ and $\bW^\vee_n$ can be obtained from one another directly by the following explicit relations that avoid combinatorics of summation over graphs
\begin{align}\label{eq:bWtobWvee}
\bW^{\vee}_n(z_{\set n},v_{\set n})
&=(-1)^n\prod_{i=1}^n
\biggl(e^{-v_ix_i}\sum_{r=0}^\infty \bigl(-d_i\tfrac{1}{dy_i}\bigr)^{r}e^{v_ix_i}[u_i^r]\biggr)
    \;\bW_n(z_{\set n},u_{\set n})+\delta_{n,1}\frac{dy_1}{v_1\hbar},
\\\label{eq:bWveetobW}
\bW_n(z_{\set n},u_{\set n})
&=(-1)^n\prod_{i=1}^n
\biggl(e^{-u_iy_i}\sum_{r=0}^\infty \bigl(-d_i\tfrac{1}{dx_i}\bigr)^{r}e^{u_iy_i}[v_i^r]\biggr)
    \;\bW^\vee_n(z_{\set n},v_{\set n})+\delta_{n,1}\frac{dx_1}{u_1\hbar}.
\end{align}

\item The (half-)differentials $\Omega_n$ and $\Omega^\vee_n$ are related to each other by integral transforms, defined by an asymptotic expansion of Gaussian integrals (\cite[Theorem 4.3]{ABDKS3}). 

\item %4. 
The $x-y$ duality preserves KP integrability (\cite[Theorem 2.3]{ABDKS3}). Namely, the system of differentials $\{\omega^{(g)}_n\}$ possesses the KP integrability property if and only if the differentials $\{\omega^{\vee,(g)}_n\}$ possesses the KP integrability property. More explicitly, denote
\begin{equation}
\bK(z^+,z^-)=\Omega_1(z^+,z^-),\qquad
\bK^\vee(z^+,z^-)=\Omega^\vee_1(z^+,z^-).
\end{equation}
Then the following determinantal relations for the differentials $\omega^{(g)}_n$ hold true
\begin{align}
\sum_{g=0}^\infty\hbar^{2g-2+n}\omega^{(g)}_n&=
(-1)^{n-1}\!\!\sum_{\sigma\in\{n\text{-cycles}\}}\prod_{i=1}^n\bK(z_i,z_{\sigma(i)}),\qquad n\ge2,
\\\sum_{g=1}^\infty\hbar^{2g-1}\omega^{(g)}_1&=\lim_{z_2\to z_1}
   \Bigl(\bK(z_1,z_2)-\tfrac{\sqrt{dx_1\,dx_2}}{x_1-x_2}\Bigr)
\end{align}
if and only if the dual differentials satisfy similar determinantal relations
\begin{align}
\sum_{g=0}^\infty\hbar^{2g-2+n}\omega^{\vee,(g)}_n&=
(-1)^{n-1}\!\!\sum_{\sigma\in\{n\text{-cycles}\}}\prod_{i=1}^n\bK^\vee(z_i,z_{\sigma(i)}),\qquad n\ge2,
\\\sum_{g=1}^\infty\hbar^{2g-1}\omega^{\vee,(g)}_1&=\lim_{z_2\to z_1}
   \Bigl(\bK^\vee(z_1,z_2)-\tfrac{\sqrt{dy_1\,dy_2}}{y_1-y_2}\Bigr).
\end{align}

Moreover, in the KP integrable case the differentials $\Omega_n$ and their $x-y$ dual differentials $\Omega^\vee_n$ admit determinantal presentations
\begin{align}\label{eq:Omega-Kdet}
\Omega_n(z^+_{\llbracket n \rrbracket},z^-_{\llbracket n \rrbracket})
 &=(-1)^{n-1} \sum_{\sigma\in \{n\text{-cycles}\}} \prod_{i=1}^n \Omega_1(z^+_{i},z^-_{\sigma(i)}), & n\geq 1,
\\\label{eq:Omega-vee-Kdet}
\Omega^\vee_n(z^+_{\llbracket n \rrbracket}, z^-_{\llbracket n \rrbracket})
 &=(-1)^{n-1} \sum_{\sigma\in \{n\text{-cycles}\}} \prod_{i=1}^n \Omega^\vee_1(z^+_{i}, z^-_{\sigma(i)}), & n\geq 1.
\end{align}

\item %5. 
Let $p$ be a zero of $dx$ which is not a zero of $dy$. Then all differentials $\omega^{\vee,(g)}_n$, $2g-2+n\geq 0$, are regular at $p$ if and only if
%the principal parts of  
the differentials $\omega^{(g)}_n$ satisfy the loop equations at $p$.

Vice versa, let $p$ be a zero of $dy$ which is not a zero of $dx$. Then all differentials $\omega^{(g)}_n$, $2g-2+n\geq 0$, are regular at $p$ if and only if the differentials $\omega^{\vee,(g)}_n$  satisfy the loop equations at $p$.

\item \label{en:xydualTR} 
 
The $x-y$ duality preserves (logarithmic) topological recursion. Namely, let $dx$ and $dy$ have non-intersecting sets of zeros, all of which are simple, and $\omega^{(0)}_2=\omega^{\vee,(0)}_2$ be the Bergman kernel, that is, it has no singularities away from the diagonal and its ${\mathfrak A}$-periods are vanishing for some choice of collection of $({\mathfrak A},{\mathfrak B})$ cycles on~$\Sigma$. If the functions $x$ and $y$ themselves are meromorphic, then the differentials $\omega^{(g)}_n$ solve topological recursion for the spectral curve data $(x,y)$ if and only if the differentials $\omega^{\vee,(g)}_n$ solve topological recursion for the spectral curve data $(y,x)$.

More generally, if the differentials $dx$ and $dy$ are meromorphic, then the differentials $\omega^{(g)}_n$ solve logarithmic topological recursion for the spectral curve data $(x,y)$ if and only if the differentials $\omega^{\vee,(g)}_n$ solve logarithmic topological recursion for the spectral curve data $(y,x)$.

\end{enumerate}
\end{theorem}

%\mkaz{It worth to mention some more useful relations:
%\begin{itemize}
%\item determinantal expressions for $\Omega_n$ and $\Omega^\vee_n$ in KP integralbe case;
%\item those relating $\Omega_n$ and $\Omega^\vee_n$ by means of Gaussian integrals.
%\end{itemize}}

\begin{example}\label{ex:(logz,z)}
We demonstrate here one relatively simple but quite shining example of $x-y$ duality. More examples and applications can be found in~\cite{ABDKS1,ABDKS3,ABDKS4}. Consider a genus zero spectral curve with the spectral curve data
\begin{equation}
x=z,\qquad y=\log z.
\end{equation}
Both $dx=dz$ and $dx^\vee=dy=\frac{dz}{z}$ have no zeros and both topological recursions for this spectral curve and its dual are trivial. However, in the presence of logarithmic singularities, it is more natural to consider the logarithmic topological recursion. The LogTR differentials for this spectral curve and its dual are given by
\begin{align}\label{eq:(logz,z)}
\omega^{(g)}_1&=\Bigl([\hbar^{2g}]\tfrac{1}{\cS(\hbar\partial_z)}\log z\Bigr)\;dz,\\
\omega^{\vee,(0)}_1&=z\; d\log z,\\
\omega^{(0)}_2&=\omega^{\vee,(0)}_2=\frac{dz_1dz_2}{(z_1-z_2)^2},
\end{align}
and all other differentials $\omega^{(g)}_n$ and $\omega^{\vee,(g)}_n$ are equal to zero (here we omit the pre-superscript ${}^{\log{}}$ in the notation). Respectively, analogously to~\eqref{eq:Wtrivz}--\eqref{eq:Wtrivlogz} we obtain explicitly
\begin{align}
\bW_n&=\prod_{i=1}^n \Bigl(e^{u_i\bigl(\frac{\cS(u_i\hbar\partial_{z_i})}{\cS(\hbar\partial_{z_i})}-1\bigr)\log z_i}dz_i\Bigr)
\;(-1)^{n-1}\!\!\sum_{\sigma\in\{n\text{-cycles}\}}\prod_{i=1}^n\frac{1}{z_i+\frac{u_{i}\hbar}{2}-z_{\sigma(i)}+\frac{u_{\sigma(i)}\hbar}{2}},
\\\bW^\vee_n&=\prod_{i=1}^n \Bigl(e^{v_i(\cS(v_i\hbar)-1)z_i}dz_i\Bigr)
\;(-1)^{n-1}\!\!\sum_{\sigma\in\{n\text{-cycles}\}}\prod_{i=1}^n\frac{1}{e^{\frac{v_{i}\hbar}{2}}z_i-e^{-\frac{v_{\sigma(i)}\hbar}{2}}
z_{\sigma(i)}}.
\end{align}
These differentials are related to one another by~\eqref{eq:bWtobWvee}--\eqref{eq:bWveetobW}, which are nontrivial identities even for this elementary case.
\end{example}

\section{Symplectic duality} \label{sec:SymplecticDuality}

\subsection{Definition of symplectic duality}

Consider a spectral curve data $(\Sigma,x,y)$ and a function $\psi(\theta)$ defined on the rational curve. Define a new spectral curve data $(x^\dag,y^\dag)$ by
\begin{equation}
x^\dag(z)=x(z)+\psi(y(z)),\qquad y^\dag(z)=y(z),
\end{equation}
where $z$ is a point on $\Sigma$.
This correspondence is a duality in a sense that the inverse transformation is produced by the dual function $\psi^\dag(\ta)=-\psi(\ta)$,
\begin{equation}
 x(z)=x^\dag(z)+\psi^\dag(y(z)).
\end{equation}

The symplectic duality construction presented below involves not only the function $\psi$ but also some $\hbar$-deformation $\tilde\psi(\ta,\hbar)$ of it which is a power series in~$\hbar^2$ such that $\tilde\psi(\ta,0)=\psi(\ta)$. The analytic properties of $\psi$ and $\tilde\psi$ will be discussed later. Basically, for the following definition to make sense we require that $\psi'(y(z))$ and $[\hbar^{2g}]\tilde\psi(y(z),\hbar)$ for $g\ge1$ are meromorphic.

\begin{definition}
Given a system of symmetric differentials $\{\omega^{(g)}_n\}$ such that $\omega^{(0)}_1=y\,dx$, the symplectic duality correspondence produces a new system of symmetric differentials denoted by $\{\omega^{\dag,(g)}_n\}$ and given by the following relations:
\begin{align}
\omega^{\dag,(0)}_1&=y\,dx^\dag,
\\ \label{eq:bW-to-omegad}
\omega^{\dag,(g)}_n(z_{\set n})&=[\hbar^{2g-2+n}]U_1\dots U_n \bW_n(z_{\set n},u_{\set n}),\quad n\ge2,
\end{align}
where $\bW_n$ is the extended differential of Definition~\ref{def:OmegabW} and $U_i$ is the operator acting on differentials in~$z=z_i$ depending on an additional variable~$u=u_i$ in a polynomial way defined by
\begin{align}\label{eq:U}
UF&=\sum_{j,r\ge0}\bigl(-d\tfrac{1}{dx^\dag}\bigr)^j[v^j]
\Bigl(\restr{\theta}{y}e^{-v\psi(\theta)}\partial_\theta^r e^{v\cS(v\hbar\partial_\theta)\tilde\psi(\theta)}\Bigr)[u^r]F(z,u)
\\ \notag &=\sum_{j,r\ge0}\bigl(d\tfrac{1}{dx^\dag}\bigr)^j[v^j]
\Bigl(\restr{\theta}{y}e^{v\psi(\theta)}\partial_\theta^r e^{-v\cS(v\hbar\partial_\theta)\tilde\psi(\theta)}\Bigr)[u^r]F(z,u).
\end{align}
In the case $n=1$ the formula is a bit more involved:
\begin{align} \label{eq:bW-to-omegad-n1}
	\omega^{\dag,(g)}_1(z_1)&= 
	[\hbar^{2g-1}] U_1 \bW_1(z_1,u_1) + [\hbar^{2g}]\sum_{j=1}^\infty \bigl(d_1\tfrac{1}{dx^\dag_1}\bigr)^{j-1} [v^j]
	\left(
	e^{-v(\cS(v\hbar\partial_{y_1})\tilde \psi(y_1)-\psi(y_1))}
	dy_1
	\right).
\end{align}
\end{definition}

\subsection{Example: hypergeometric tau functions and associated differentials}

Let $\psi(\ta)$ and $y(z)$ be some formal power series such that $\psi(0)=0$, $y(0)=0$. Assume that we are given also some $\hbar^2$-deformations $\hat\psi(\ta,\hbar)$ and $\hat y(z,\hbar)$ of these functions. These data determine the so-called Orlov--Scherbin (or, in other terminology, hypergeometric) tau functions~\cite{OS} and the corresponding $n$-point differentials~\cite{BDKS1,BDKS2}. %\mkaz{(Give precise definitions through either Schur function expansion or VEVs.)}

\begin{definition} The Orlov--Scherbin (or hypergeometric) tau function
	%, that generates all OS differentials, 
	is given by the sum over partitions
	\begin{equation}\label{eq:OSpartfun}
		Z^{\mathrm{OS}}=\sum_{\lambda}\shin_\lambda(t)\shin_\lambda(\hbar^{-1} s)\exp\left(\sum_{(i,j)\in \lambda}\hat\psi(\hbar(i-j),\hbar)\right),
	\end{equation}
	where $\shin_\lambda(t)$ is the Schur function in the variables $t=(t_1,t_2,\dots)$ labeled by the partition $\lambda$ and the parameters $s_i=s_i(\hbar)$ for the second Schur function are the expansion coefficients defined by $\hat y(z,\hbar)=\sum_{j=1}^\infty s_j(\hbar)z^j$.
\end{definition}

\begin{definition} The Orlov--Scherbin differentials are defined as 
\begin{align}
	\sum_{g=0}^\infty \hbar^{2g-2+n} \omega^{\mathrm{OS},(g)}_n = \prod_{j=1}^\infty \restr{t_j}{0}\prod_{i=1}^n \bigg(d_i \sum_{k_i=1}^\infty X_i^{k_i} \partial_{t_{k_i}} \bigg)\log Z^{\mathrm{OS}}. 	
\end{align}
\end{definition}

Under some natural analyticity assumptions (see \cite[Definition 1.4]{BDKS2}) on the initial data, the Orlov--Scherbin  differentials extend as global meromorphic $n$-differentials on the spectral curve $\Sigma=\C \mathrm{P}^1$ and we have
\begin{align}
X&=z\,e^{-\psi(y(z))},
\\\omega^{{\rm OS},(0)}_1&=y\frac{dX}{X}.
\end{align}

Let us denote
\begin{equation}
x=\log X=\log z-\psi(y(z)),\quad x^\dag=\log z,
\quad \tilde\psi(\ta)=\frac{1}{\cS(\hbar\partial_\ta)}\hat\psi(\ta).
\end{equation}
Then we have $\omega^{{\rm OS},(0)}_1=y\,dx$ and the following proposition is a reformulation of \cite[Proposition~2.3]{BDKS2}.

\begin{proposition}\label{prop:OSsympl}
The Orlov--Scherbin  differentials $\omega^{(g)}_n=\omega^{{\rm OS},(g)}_n$ can be computed as symplectic dual to the %trivial 
system of differentials
\begin{equation}\label{eq:OSDualDiff}
\omega^{\dag,(g)}_n=\delta_{n,1}[\hbar^{2g}]\hat y_1dx_1^\dag+\delta_{(g,n),(0,2)}\tfrac{dz_1dz_2}{(z_1-z_2)^2}
\end{equation}
%associated with the spectral curve $(\log z,y(z))$ by the LogTR, 
%and computed 
with respect to the function $\tilde\psi^\dag(\ta,\hbar)=-\tilde\psi(\ta,\hbar)=-\frac{1}{\cS(\hbar\partial_\ta)}\hat\psi(\ta,\hbar)$.
\end{proposition}

More explicitly, the formula for the differentials of Orlov--Scherbin  family for $g\geq 0$, $n\geq 1$ reads
\begin{align}\label{eq:OS}
\omega^{(g)}_n & =[\hbar^{2g-2+n}]
\prod_{i=1}^n\biggl(
\sum_{j,r\ge0}\bigl(d_i\tfrac{1}{dx_i}\bigr)^j
\Bigl(\restr{\theta}{y_i}[v^j]e^{-v\psi(\theta)}\partial_\theta^r e^{v\cS(v\hbar\partial_\theta)\tilde\psi(\theta)}\Bigr)[u_i^r]
\biggr)
\\ \notag & \qquad 
\prod_{i=1}^n\Bigl(e^{u_i(\cS(u_i\hbar\,z_i\partial_{z_i})\hat y_i-y_i)}dz_i\Bigr) \cdot
 (-1)^{n-1}\!\!\sum_{\sigma\in\{n\text{-cycles}\}}\prod_{i=1}^n\frac{1}{e^{\frac{u_{i}\hbar}{2}}z_i-e^{-\frac{u_{\sigma(i)}\hbar}{2}}
    z_{\sigma(i)}}
	\\ \notag & \quad+
\delta_{n,1}[\hbar^{2g}]\sum_{j=0}^\infty \bigl(d_1\tfrac{1}{dx_1}\bigr)^{j-1} [v^j]
\left(
e^{v(\cS(v\hbar\partial_{y_1})\tilde \psi(y_1)-\psi(y_1))}
dy_1
\right).
\end{align}

For $\hat y$ we use the notation $\hat y_i := \hat y (z_i,\hbar)$ too.

\begin{proof}[Proof of Prop.~\ref{prop:OSsympl}] We refer the reader to \cite[Proposition~2.3]{BDKS2}. The only difference between the expression given there and Equation~\eqref{eq:OS} is that the former uses the following expression instead of the second factor in the second line of~\eqref{eq:OS}:
\begin{align}
	\prod_{i=1}^n\frac{1}{u_i\hbar\cS(u_i\hbar)z_i} \sum_{\gamma\in\Gamma_n}\prod_{(v_k,v_l)\in E_\gamma} w_{k,l},
\end{align}
where the sum is taken over all simple connected graphs $\gamma$ on $n$ labeled vertices, and the edge $(v_k,v_l)$ in the set of edges $E_\gamma$ of $\gamma$ is decorated by 
\begin{align}
	w_{k,\ell}&\coloneqq e^{\hbar^2u_ku_l\cS(u_k\hbar Q_k D_k)\cS(u_\ell\hbar Q_\ell D_\ell)\frac{z_k z_\ell}{(z_k-z_\ell)^2}}-1.
\end{align}
The identification of this expression with 
\begin{align}
	(-1)^{n-1}\!\!\sum_{\sigma\in\{n\text{-cycles}\}}\prod_{i=1}^n\frac{1}{e^{\frac{u_{i}\hbar}{2}}z_i-e^{-\frac{u_{\sigma(i)}\hbar}{2}}
		z_{\sigma(i)}}
\end{align}
 is the connected version of the Cauchy determinant formula
\begin{align}
(-1)^{n-1}\!\!\sum_{\sigma\in\{n\text{-cycles}\}}\prod_{i=1}^n\frac{1}{a_i-b_{\sigma(i)}}
=\prod_{i=1}^n\frac{1}{a_i-b_i}\;\sum_{\gamma\in\Gamma_n}\prod_{(v_k,v_\ell)\in E_\gamma}\Bigl(\tfrac{(a_k-a_\ell)(b_k-b_\ell)}{(a_k-b_\ell)(b_k-a_\ell)}-1\Bigr)	
\end{align}
applied to $a_k=e^{\frac{u_k\hbar}{2}}z_k=e^{\frac{u_k\hbar}{2}z_k\partial_{z_k}}z_k$,  
$b_k=e^{-\frac{u_k\hbar}{2}}z_k=e^{-\frac{u_k\hbar}{2}z_k\partial_{z_k}}z_k$.
	
	In the case $n=1$, $g\ge1$ the expression for $\omega^{(g)}_n$ contains an additional term, see \cite[Equation~(32)]{BDKS2}, which matches with the last summand in~\eqref{eq:OS}.
\end{proof}

\subsection{Example: \texorpdfstring{$(t,s)$}{(t,s)}-weighted hypermaps}

To a given series $\phi(\ta)$ such that $\phi(0)=1$ one associates a tau function depending on two sets of times $t=(t_1,t_2,\dots)$ and $s=(s_1,s_2,\dots)$ and a two-index set of $(m,n)$-point functions $H^{(g)}_{m,n}(X_1,\dots,X_m,Y_1,\dots,Y_n)$ (depending also on $(t,s)$-parameters). In the case $\phi(\ta)=1+\ta$ these functions enumerate genus~$g$ hypermaps with marked and unmarked vertices of two colors known also as cumulants of 2-matrix model, see~\cite[Introduction]{BDKS4}.

More precisely, we use the following modification of the construction given in the previous section. We consider  
\begin{equation}
	Z=\sum_{\lambda}\shin_\lambda(\hbar^{-1} t)\shin_\lambda(\hbar^{-1} s)\prod_{(i,j)\in \lambda}\phi(\hbar(i-j)),
\end{equation}
for $\phi(\theta)\coloneqq 1+\theta$ and we define $H^{(g)}_{m,n}$ as 
\begin{align}
	\sum_{g=0}^\infty \hbar^{2g-2} H^{(g)}_{m,n} = \prod_{i=1}^m \bigg( \sum_{k_i=1}^\infty X_i^{k_i} \partial_{t_{k_i}} \bigg)\prod_{i=1}^n \bigg( \sum_{\ell_i=1}^\infty Y_i^{\ell_i} \partial_{s_{\ell_i}} \bigg)\log Z. 	
\end{align}
For a general $\phi$ they enumerate certain weighted hypermaps or weighted Hurwitz numbers, see e.g.~\cite{ALS,Harnad}.
%For a general $\phi$ it enumerates certain weighted hypermaps.

Consider the case when $\phi$ is a polynomial and in $H^{(g)}_{m,n}$ we set all but finitely many $(t,s)$ parameters to zero. % only finitely many $(t,s)$ parameters are distinct from zero. 
Under these assumptions, one constructs in~\cite{BDKS4} a rational spectral curve $\Sigma = \C\mathrm{P}^1$ equipped with three rational functions $X,Y,\Theta$ satisfying $X\,Y\,\phi(\Theta)=1$ such that $X$ is a local coordinate at the point $z=0$, and $Y$ is a local coordinate at the point $z=\infty$. It is proved that the function $H^{(g)}_{m,n}$ introduced initially as a power series in $X,Y$ variables extends as a global rational function on $\Sigma^{m+n}$.
Denote
\begin{align}
x&=\log X,\qquad x^\dag=-\log Y,\qquad y=\Theta,
\\\psi(\ta)&=\log\phi(\ta),%=x^\dag-x=\log\phi,
\\\omega^{(0)}_1&=\Theta\frac{dX}{X}=y\,dx,\qquad
\omega^{\dag,(0)}_{1}=-\Theta\frac{dY}{Y}=y\,dx^\dag,
\\\omega^{(0)}_2&=\omega^{\dag\,(0)}_2=\frac{dz_1dz_2}{(z_1-z_2)^2},
\\\omega^{(g)}_n&=d_1\dots d_n H^{(g)}_{n,0},\quad 2g-2+n>0,
\\\omega^{\dag,(g)}_n&=(-1)^n d_1\dots d_n H^{(g)}_{0,n},\quad 2g-2+n>0.
\end{align}

With this notation, the results of~\cite{BDKS4} expressing $H^{(g)}_{n,0}$ via  $H^{(h)}_{0,m}$ can be reformulated in the following way

\begin{theorem}
The differentials $\omega^{(g)}_n$, $\omega^{\dag,(g)}_n$
are symplectic dual for the spectral curves $(\log X,\Theta)$ and $(-\log Y,\Theta)$ and computed with respect to the $\hbar$-deformation $\tilde\psi(\ta)=\frac{1}{\cS(\hbar \partial_\ta)}\log\phi(\ta)$ of the function $\log\phi$. %=x^\dag-x$.
\end{theorem}

\begin{remark} In fact, no explicit expression of $H^{(g)}_{n,0}$ via $H^{(h)}_{0,m}$ is given explicitly in~\cite{BDKS4}. In \emph{loc.~cit.} only the formulas expressing $H^{(g)}_{m,n+1}$ through $H^{(g')}_{m',n'}$ with $n'\leq n$ are given. But then one can combine these expressions into a formula expressing $H^{(g)}_{n,0}$ via  $H^{(h)}_{0,m}$.   
\end{remark}

\subsection{Example: \texorpdfstring{$\psi=\log(\ta)$}{psi=log(theta)} symplectic duality as the original formulation of \texorpdfstring{$x-y$}{x-y} duality}

Consider the initial differentials of $x-y$ duality
\begin{equation}
-\omega^{(0)}_1=-y\,dx=xy\;d\log x^{-1},
\qquad \omega^{\vee,(0)}_1=x\,dy=xy\;d\log y.
\end{equation}
We see that these differentials can also be treated as initial differentials for the spectral curves $(\log x^{-1},x\,y)$ and $(\log y,x\,y)$, respectively. Note that $\log y-\log x^{-1}=\psi(x\,y)$ with $\psi(\ta)=\log \ta$. This motivates the following result.

\begin{proposition} The dual differentials $\omega^{\vee,(g)}$ of $x-y$ duality coincide with the symplectic dual of the differentials $(-1)^n\omega^{(g)}_n$ computed with respect to symplectic dual spectral curves $(\log x^{-1},x\,y)$ and $(\log y,x\,y)$ and the $\hbar$-deformation \begin{align}
\tilde\psi(\ta)=\frac{1}{\cS(\hbar\partial_\ta)}\log \ta
	\end{align} of the function $\psi(\ta)=\log \ta$.
\end{proposition}

This proposition is proved in~\cite[Section~3]{ABDKS1}. This treatment of $x-y$ duality appeared in the literature as its first known formula, see~\cite{borot2023functional,BDKS-fullysimple,ABDKS1}. Note that this definition can only be applied if both $x$ and $y$ are meromorphic. Besides, it is not invariant under a change of $x$ or~$y$ by additive constants. Meanwhile, for the correctness of Definition~\ref{def:xy-duality} it is sufficient to assume that $dx$ and $dy$ are meromorphic and the duality transformation is uniquely determined by these differentials.

\section{Symplectic duality as a combination of \texorpdfstring{$x-y$}{x-y} dualities }

\label{sec:SymplDual-via-xy}

Consider the spectral curves $x-y$ dual to the original one given by the pair of functions $(x,y)$ and its symplectic dual one $(x^\dag,y^\dag)=(x+\psi(y),y)$,
\begin{align}\label{eq:checkdef}
(x^\vee,y^\vee)&=(y,x),
& &\omega^{\vee,(0)}_1=x\,dy,
\\ \label{eq:dagcheckdef}
(x^{\dag\vee},y^{\dag\vee})&=(y,x+\psi(y)),
& &\omega^{\dag\vee,(0)}_1=(x+\psi(y))\,dy.
\end{align}
These curves are very close. For example, if $x$ and $y$ are meromorphic and $\tilde\psi=\psi$ is rational, then the topological recursion differentials  produced by these spectral curves are the same for $(g,n)\ne(0,1)$. This provides a hint that the $x-y$ dual differentials of the original system of differentials and its symplectic dual should be related in a more general formal setting. The following theorem claims that it is indeed the case.

\begin{theorem}\label{thm:sympl-from-xy}
Consider a given collection of differentials $\{\omega^{(g)}_n\}$ and its symplectic dual one $\{\omega^{\dag,(g)}_n\}$ computed with respect to a chosen function $\tilde\psi$. Then the $x-y$ dual of these two systems of differentials, $\{\omega^{\vee,(g)}_n\}$, and $\{\omega^{\dag\vee,(g)}_n\}$, respectively, are related to one another by
\begin{equation}\label{eq:omegav-omegadv}
\omega^{\dag\vee,(g)}_n=\omega^{\vee,(g)}_n+\delta_{n,1}\Bigl([\hbar^{2g}]\tilde\psi(y(z),\hbar)\Bigr)\,dy.
\end{equation}

In other words, 
%Theorem claims that 
the symplectic dual differentials can be obtained by the following combination of three transformations:
\begin{enumerate}[label=(\alph*)]
\item \label{en:trxy} apply $x-y$ duality to the original system of differentials;
\item \label{en:trShift1} modify $(g,1)$-differentials of the obtained system by the summands prescribed by~\eqref{eq:omegav-omegadv};
\item \label{en:trxyback} apply $x-y$ duality once again.
\end{enumerate}
\end{theorem}

\begin{proof}
In fact, while applying these transformations, there is no need to involve combinatorics of summation over graphs each time. It is much more instructive to follow the action of these transformations on the corresponding extended differentials which, according to Definition \ref{def:OmegabW}, already contain these summations.

 In particular, Equation~\eqref{eq:omegav-omegadv} rewritten in terms of extended differentials reads
\begin{equation}\label{eq:bWv-bWdv}
\bW^{\dag\vee}_n(z_{\set n},v_{\set n})=\bW^{\vee}_n(z_{\set n},v_{\set n})
\prod_{i=1}^n e^{v_i(\cS(v_i\hbar\partial_{y_i})\tilde \psi(y_i)-\psi(y_i))}.
\end{equation}
Substituting to~\eqref{omegaveeW}--\eqref{eq:bWtobWvee} we obtain the following expression for the dual differentials for $n\ge2$ suggested by Theorem~\ref{thm:sympl-from-xy}
\begin{equation}
\omega^{\dag,(g)}_n=[\hbar^{2g-2+n}]\prod_{i=1}^n\Bigl(
\sum_{r,j\ge0}
\bigl(-d_i\tfrac{1}{dx^\dag_i}\bigr)^j
[v_i^j]
e^{v_i(\cS(v_i\hbar\partial_{y_i})\tilde\psi(y_i)-\psi(y_i))}
e^{-v_i x_i}\bigl(-d_i\tfrac{1}{dy_i}\bigr)^r e^{v_i x_i}[u_i^r]
\Bigr)\bW_n,%(z_{\set n},u_{\set n}
\end{equation}
and expressions for $\omega^{\dag,(g)}_1$ differentials contain correction which are easy to trace. 

Comparing with the actual definition~\eqref{eq:bW-to-omegad} of the differentials $\omega^{\dag,(g)}_n$ (and its version for the special case $n=1$ given in~\eqref{eq:bW-to-omegad-n1}) we see that Theorem follows from the following operator identity: for any meromorphic differential $F(z,u)$ depending on~$u$ in polynomial way we have
\begin{align}\label{eq:id-xy-sympl}
UF & =\sum_{r,j\ge0}\bigl(-d\tfrac{1}{dx^\dag}\bigr)^j[v^j]
\Bigl(\restr{\theta}{y}e^{-v\psi(\theta)}\partial_\theta^r e^{v\cS(v\hbar\partial_\theta)\tilde\psi(\theta)}\Bigr)[u^r]F(z,u)
\\ \notag & =\sum_{r,j\ge0}
\bigl(-d\tfrac{1}{dx^\dag}\bigr)^j
[v^j]
e^{v(\cS(v\hbar\partial_{y})\tilde\psi(y)-\psi(y))}
e^{-v x}\bigl(-d\tfrac{1}{dy}\bigr)^re^{v x}[u^r]F(z,u).
\end{align}
We claim that not only the whole expressions on both sides but also their summands for each particular $r$ agree. Namely, denote
\begin{equation}\label{eq:ABdef}
A(z,v)=e^{-v\psi(y)}[u^r]\frac{F(z,u)}{dy},
\qquad
B(z,v)=e^{v\cS(v\hbar\partial_y)\tilde\psi(y)}.
\end{equation}
With this notation and using equalities $x+\psi(y)=x^\dag$, $\bigl(d\tfrac{1}{dx}\bigr)^j=dx\, \partial_x^{j}\frac{1}{dx}$, we can write the equality of the $r$th terms in~\eqref{eq:id-xy-sympl} that we are going to prove in the following form
\begin{equation}\label{eq:id-xy-sympl2}
\sum_{j=0}^\infty(-\partial_{x^\dag})^j\;\frac{dy}{dx^\dag}[v^j]A\partial_y^r B=
\sum_{j=0}^\infty(-\partial_{x^\dag})^j\;\frac{dy}{dx^\dag}[v^j]B e^{-vx^\dag}(-\partial_y)^r e^{vx^\dag}A.
\end{equation}

\begin{lemma}\label{lem:ABlemma}
Identity~\eqref{eq:id-xy-sympl2} holds for any $r\geq 0$ and for any functions $x^\dag(z)$, $y(z)$, $A(z,v)$, $B(z,v)$, where~$A$ and~$B$ considered as functions in~$v$ are power series such that the products $A\,B$ and $A\partial_z B$ are polynomial in $v$.
\end{lemma}

\begin{proof}
Indeed, using that
\begin{multline}
0=\sum_{j=0}^\infty(-\partial_{x^\dag})^j[v^j](\partial_{x^\dag}+v)A\,B
=\sum_{j=0}^\infty(-\partial_{x^\dag})^j[v^j]e^{-vx^\dag}\,\partial_{x^\dag}\,e^{vx^\dag}A\,B
\\=\sum_{j=0}^\infty(-\partial_{x^\dag})^j\frac{dy}{dx^{\dag}}[v^j]e^{-vx^\dag}\,\partial_y\,e^{vx^\dag}A\,B
=\sum_{j=0}^\infty(-\partial_{x^\dag})^j\frac{dy}{dx^{\dag}}[v^j](A\,\partial_y B+e^{-vx^{\dag}}B\,\partial_y(e^{vx^{\dag}}A))
\end{multline}
we conclude that
\begin{equation}
\sum_{j=0}^\infty(-\partial_{x^\dag})^j\frac{dy}{dx^{\dag}}[v^j]A\,\partial_y B
=\sum_{j=0}^\infty(-\partial_{x^\dag})^j\frac{dy}{dx^{\dag}}[v^j]e^{-vx^{\dag}}B\,(-\partial_y)(e^{vx^{\dag}}A).
\end{equation}
Applying this relation $r$ times we obtain the equality of Lemma.
\end{proof}

Note that $A$ and $B$ of~\eqref{eq:ABdef} satisfy the conditions of Lemma~\ref{lem:ABlemma} \emph{order by order in $\hbar$}. Thus, we can apply this lemma (order by order in $\hbar$), which proves%This Lemma proves
~\eqref{eq:id-xy-sympl} and hence Theorem~\ref{thm:sympl-from-xy}.
\end{proof}
%\begin{remark}
%By polynomiality here we mean that each coefficient of the $\hbar$ expansion of the expression is a polynomial.
%\end{remark}

As an immediate corollary of Theorem~\ref{thm:sympl-from-xy} we conclude that the symplectic duality transformation shares most of the properties valid for $x-y$ duality, see Theorem \ref{Th:x-y}.

\begin{theorem}
1. The symplectic dual differentials $\omega^{\dag,(g)}_n$ are meromorphic, symmetric, and possess the property that $\omega^{\dag,(g)}_n-\delta_{(g,n),(0,2)}\frac{dx^\dag_1dx^\dag_2}{(x^\dag_1-x^\dag_2)^2}$ are regular on the diagonals.

2. This correspondence is a duality indeed: the original differentials $\omega^{(g)}_n$ can be recovered from the symplectic dual ones by a dual formula that in the case $n\ge 2$ reads
\begin{align}
\omega^{(g)}_n(z_{\set n})
=&[\hbar^{2g-2+n}]U^\dag_1\dots U^\dag_n \bW^\dag_n(z_{\set n},w_{\set n}),
\end{align}
where
\begin{align}
U^\dag F&=
\sum_{j,r\ge0}\bigl(-d\tfrac{1}{dx}\bigr)^j[v^j]
\Bigl(\restr{\theta}{y}e^{v\psi(\theta)}\partial_\theta^r e^{-v\cS(v\hbar\partial_\theta)\tilde\psi(\theta)}\Bigr)[w^r]F(z,w),
\end{align}
and for $n=1$ we have
\begin{align}
	\omega^{\dag,(g)}_1(z_1)&= 
	[\hbar^{2g-1}] U^\dag_1 \bW^\dag_1(z_1,u_1) + [\hbar^{2g}]\sum_{j=1}^\infty \bigl(d_1\tfrac{1}{dx_1}\bigr)^{j-1} [v^j]
	\left(
	e^{v(\cS(v\hbar\partial_{y_1})\tilde \psi(y_1)-\psi(y_1))}
	d_1y_1
	\right).
\end{align}

3. Moreover, the extended dual differentials  $\bW_n$ and $\bW^\dag_n$ can be obtained from one another directly by the following explicit relations that avoid combinatorics of summation over graphs
\begin{align}\label{eq:bWtobWdag}
\bW^{\dag}_n(z_{\set n},w_{\set n})
&=\tilde U_1\dots \tilde U_n\bW(z_{\set n},u_{\set n}) + \delta_{n,1} \frac{dx^\dag_1}{\hbar w_1},
\\\tilde U F&=e^{-w y}\sum_{j,r\ge0}\bigl(-d\tfrac{1}{dx^\dag}\bigr)^j[v^j]
\Bigl(\restr{\theta}{y}e^{-v\psi(\theta)}\partial_\theta^r e^{v\cS(v\hbar\partial_\theta)\tilde\psi(\theta)+w\theta}\Bigr)[u^r]F(z,u),
\\\label{eq:bWdadtobW}
\bW_n(z_{\set n},u_{\set n})
&=\tilde U_1^\dag\dots\tilde U_n^\dag\bW(z_{\set n},w_{\set n}) + \delta_{n,1} \frac{dx_1}{\hbar u_1},
\\\tilde U^\dag F&=
e^{-u y}\sum_{j,r\ge0}\bigl(-d\tfrac{1}{dx}\bigr)^j[v^j]
\Bigl(\restr{\theta}{y}e^{v\psi(\theta)}\partial_\theta^r e^{-v\cS(v\hbar\partial_\theta)\tilde\psi(\theta)+u y}\Bigr)[w^r]F(z,w).
\end{align}

The differentials $\Omega_n$ and $\Omega^\dag_n$ are related to each other by integral transforms, defined by an asymptotic expansion around Gaussian integrals. For each of $2n$ variables a transform is given by a double integral (instead of single one in the $x-y$ swap case).

4. The symplectic duality preserves KP integrability. Namely, the system of differentials $\{\omega^{(g)}_n\}$ possesses the KP integrability property if and only if the differentials $\{\omega^{\dag,(g)}_n\}$ possesses the KP integrability property. Equivalently, the determinantal expressions for the differentials hold on one side of symplectic duality if and only if they hold on the other side of duality.

5. %\pdb{I don't think that this holds. It seems that for symplectic duality we need to consider LEs and proj.property together}. 
Assume that $\psi$ is unramified and regular at the critical values of~$y$. Let $p$ be a zero of $dx$ which is not a zero of $dy$. Then all differentials $\omega^{\dag,(g)}_n$, $(g,n)\ne(0,1)$, are regular at $p$ if and only if the differentials $\omega^{(g)}_n$ satisfy the loop equations at $p$.

Vice versa, let $p$ be a zero of $dy$ which is not a zero of $dx$. Then all differentials $\omega^{(g)}_n$, $(g,n)\ne(0,1)$, are regular at $p$ if and only if the differentials $\omega^{\dag,(g)}_n$ satisfy loop equations at $p$.

%6. The $x-y$ duality preserves (logarithmic) topological recursion. Namely, let $dx$ and $dy$ have non-intersecting sets of zeros, and $\omega^{(0)}_2=\omega^{\vee,(0)}_2$ be the Bergman kernel, that is, it has no singularities away from the diagonal and its $A$-periods are vanishing for some choice of collection of $(A,B)$ cycles on~$\Sigma$. If the functions $x$ and $y$ themselves are meromorphic, then the differentials $\omega^{(g)}_n$ solve topological recursion for the spectral curve data $(x,y)$ if and only if the differentials $\omega^{\vee,(g)}_n$ solve topological recursion for the spectral curve data $(y,x)$.
%
%More general, if the differentials $dx$ and $dy$ are meromorphic, then the differentials $\omega^{(g)}_n$ solve logarithmic topological recursion for the spectral curve data $(x,y)$ if and only if the differentials $\omega^{\vee,(g)}_n$ solve logarithmic topological recursion for the spectral curve data $(y,x)$.
\end{theorem}

All these properties are direct corollaries of Theorem~\ref{thm:sympl-from-xy} and the corresponding properties of $x-y$ duality. Part of these properties was already observed though stated a bit differently in~\cite{ABDKS2}. A slightly more delicate question on the compatibility with the projection property and its logarithmic version, and hence with topological recursion or logarithmic topological recursion will be discussed in the next section.

We finish this section with the following observation.

\begin{proposition}
The symplectic duality possesses the following group property: the composition of two symplectic duality transformations computed with respect to the given functions $\tilde\psi_1$ and $\tilde\psi_2$, respectively, is equivalent to symplectic duality transformation computed with respect to the function $\tilde\psi_1+\tilde\psi_2$.
\end{proposition}

\begin{proof}
Indeed, after conjugation with $x-y$ duality the transformation of symplectic duality results in a change of the $1$-point differentials by a summand proportional to $\tilde\psi(y)$. Such transformation obviously satisfies the group property.
\end{proof}

\section{Compatibility with topological recursion} \label{sec:CompatibilityTR}

\subsection{Topological recursion}
The following statement is a direct corollary of Theorem~\ref{thm:sympl-from-xy}.

\begin{theorem}\label{thm:TRsymplectic}
Assume that $x$ and $y$ are meromorphic and $\psi$ is rational and regular at the critical values of $y$. Let $x^\dag=x+\psi(y)$ and assume that all zeros of $dx$, $dy$, and $dx^\dagger$ are simple, and that the zeros of $dx$ are disjoint from the zeros of $dy$. %, and also that the zeros of $dx^\dag$ are disjoint from the zeros of $dy$.

Let $\tilde \psi=\psi$. Then the system of differentials $\{\omega^{(g)}_n\}$ solves TR for the spectral curve $(\Sigma,x,y)$ if and only if its symplectic dual system of differentials $\{\omega^{\dag,(g)}_n\}$ solves TR for the spectral curve $(\Sigma,x^\dag,y)$.
\end{theorem}

\begin{proof}
The notation in this proof follows the notation of Theorem~\ref{thm:sympl-from-xy} (in particular, see \eqref{eq:checkdef} and \eqref{eq:dagcheckdef}).
%Indeed, the first and the third transformations of Theorem~\ref{thm:sympl-from-xy} preserve relations of topological recursion by Theorem \ref{Th:x-y} . 
%
Note that transformation~\ref{en:trxy} of Theorem~\ref{thm:sympl-from-xy} preserves topological recursion by Theorem~\ref{Th:x-y} (part~\ref{en:xydualTR}).

Let us see that after the application of transformation~\ref{en:trShift1} of Theorem~\ref{thm:sympl-from-xy} TR still holds (this is a very well-known fact in the theory of topological recursion, see~\cite{EO}; still we discuss its proof here as it will be useful for our discussion of the case of LogTR below).

Keeping in mind that $\tilde{\psi}=\psi$, formula \eqref{eq:omegav-omegadv} says that $\omega^{\dagger,\vee,(g)}_n$ differs from $\omega^{\vee,(g)}_n$ only in $(g,n)=(0,1)$ term.
The linear loop equations still hold (that is, they hold for $\omega^{\dagger,\vee,(g)}_n$) due the fact that $\psi(\ta)$ is regular at the critical values of $x^\vee$ and thus $\psi(x^\vee(z))dx^\vee$ is regular and vanishing at the critical points of $x^\vee$.
Regarding the quadratic loop equations, let us use the form from Remark~\ref{rem:QLEDelta}, as the linear loop equations are already established.
As the change in \eqref{eq:omegav-omegadv} affects only the $n=1$ terms then only the second summand in the quadratic loop equation \eqref{eq:QLE-newform} is affected. Note that we then have either $\Delta_{i,z}$ acting on $\psi(x^\vee(z))dx^\vee$ or $\Delta_{i,w}$ acting on $\psi(x^\vee(w))dx^\vee$, but $x^\vee$ is by definition symmetric under the deck transformations and thus $\Delta$ kills $\psi(x^\vee)dx^\vee$. This means that the quadratic loop equations are preserved.
The projection property is imposed only upon the $2g-2+n>0$ differentials and thus it is preserved as well, as only $\omega^{(0)}_1$ gets changed.

Thus, after transformation~\ref{en:trShift1} TR still holds, and after transformation~\ref{en:trxyback} it holds as well again due to Theorem~\ref{Th:x-y} (part~\ref{en:xydualTR}). In the latter statement we have used the fact that the zeros of $dx^\dag$ are disjoint from the zeros of $dy$ since at the zeros of $dy$ we have $dx=dx^\dagger$ due to the condition that $\psi$ is regular at the critical values of $y$.
\end{proof}

%\mkaz{Do we need the convention $\Sigma=\C P^1$ in this section? Do we need to formulate explicitly the generality condition each time, or it is sufficient to postulate it once for all cases?}
\begin{remark}
	Let us stress that the above theorem works for an arbitrary compact Riemann surface $\Sigma$ of any genus.
\end{remark}

\subsection{Logarithmic topological recursion}

We would like to extend the result of Theorem~\ref{thm:TRsymplectic} to the situation when $x$ and $y$ may possess logarithmic singularities, that is, $dx$ and $dy$ are meromorphic but might have non-vanishing residues. Then the condition that the differential of $x^\dag=x+\psi(y)$ is also meromorphic imposes certain conditions on $x$, $y$, and $\psi$. 

In the presence of logarithmic singularities, it is more natural to consider not the traditional topological recursion but \emph{logarithmic topological recursion} defined in Section~\ref{sec:LogTRDef}, see also~\cite{ABDKS4}. The LogTR differentials comparing with TR differentials may posses a small number of extra poles with controlled principal parts. One of the main observations of~\cite{ABDKS4} is that the $x-y$ duality preserves LogTR rather than just TR. 

Due to the fact that $y$ is multivalued when it has logarithmic singularities, it is tricky to formulate and prove a general theorem in the vein of Theorem~\ref{thm:TRsymplectic} in this case, since we need to deal with the composition $\psi(y(z))$. So instead we are going to formulate a \emph{principle} for which we aim in this case.
\begin{principle}\label{pr:sympLogTR}
	Assume that $dx$ and $dy$ are meromorphic. 
	  Let $\psi$ be such that the composition $\psi(y(z))$ makes sense and $d\psi(y(z))$ is univalued and meromorphic. Let $x^\dag=x+\psi(y)$ and assume that all zeros of $dx$, $dy$, and $dx^\dagger$ are simple, and that the zeros of $dx$ are disjoint from the zeros of $dy$, and $\psi$ is regular at the critical values of $y$. 
	  
	%, and also that the zeros of $dx^\dag$ are disjoint from the zeros of $dy$. % Let us also assume that zeros (all of which are simple) of $dx$ and $dy$ are disjoint, and that zeros (all of which are simple) of $dx^\dag$ and $dy$ are disjoint as well, where $x^\dag=x+\psi(y)$.
	
	Let $\tilde \psi$ be the %unique 
	$\hbar$-deformation of $\psi$ described below. Then the system of differentials $\{\omega^{(g)}_n\}$ solves LogTR for the spectral curve $(x,y)$ if and only if its symplectic dual system of differentials $\{\omega^{\dag,(g)}_n\}$ solves LogTR for the spectral curve $(x^\dag,y)$.
\end{principle}
We do not prove the statement of this principle in general, but we prove it below for three large families of $(x,y,\psi)$ 
for which the conditions are satisfied.

\begin{remark}
	All of the three families for which we prove precise statements realizing Principle~\ref{pr:sympLogTR} below correspond to genus-zero spectral curves, i.e. to $\Sigma=\mathbb{C}\mathrm{P}^1$. The principle \emph{should} work for higher genus spectral curves as well (in particular, when $x$ and/or $y$ become multivalued not due to the presence of logarithmic singularities, but because of monodromy around non-contractible cycles, while $dx$ and $dy$ are still meromorphic), for some suitable restrictions on $x$, $y$, and $\psi$, but studying this lies outside the scope of the present paper; we hope that this will be covered in some future works.
\end{remark}

\subsubsection{Construction of \texorpdfstring{$\tilde \psi$}{psi}}\label{sec:tildepsi}

Let us assume  that $\Sigma=\mathbb{C}\mathrm{P}^1$.

The construction for $\tilde \psi$ goes as follows. Let $\{a_1,\dots,a_M\}$ be the set of all simple poles of $d\psi(y(z))$ which are not simultaneously poles of $dy(z)$. Then $[\hbar^{2g}]\tilde \psi(y(z),\hbar)$, for $g\geq 1$, is the unique function (here $\tilde \psi(y(z),\hbar)$ is treated not as a composition but just as some function depending on $z$ and $\hbar$) such that  $[\hbar^{2g}]\tilde \psi(y(z),\hbar)dy(z)$ does not have any poles apart from $a_1,\dots,a_M$, and its principal part at $a_i$, for $i=1,\dots,M$, coincides with the principal part of 	
	\begin{align}\label{eq:tildepsipp}
		[\hbar^{2g}]\left( \frac{1}{\alpha_i\cS(\alpha_i\hbar \partial_{y})}\log(z-a_i)\right)dy,
	\end{align}
where $\alpha_i^{-1}$ is the residue of $d\psi(y(z))$ at $a_i$. It is not clear how to pass from $\tilde \psi(y(z),\hbar)$ to $\tilde \psi(\ta,\hbar)$ in general, but in the cases described below it is clear. In particular, if $M=0$ then $\tilde{\psi}(\ta,\hbar)=\psi(\ta)$.

%\begin{proof}[Proof of Theorem~\ref{thm:sympLogTR}]
%\end{proof}

\subsubsection{Principle at work} \label{sec:principleAtWork}

We discuss below %provide 
three cases (families) of\\ $(x,y,\,\psi)$ satisfying the conditions of Principle~\ref{pr:sympLogTR}. We do not claim that this list is exhaustive, but the respective theorems, in particular, imply that TR holds for (generating functions of) all generalized double Hurwitz numbers appearing in the literature, as we will see below. 
In all cases we assume that $\Sigma=\mathbb{C}\mathrm{P}^1$.
	\begin{enumerate}[label=Case \Roman*.,itemsep=5pt,topsep=8pt]
		\item \label{en:caseI}
	%Case I: 
	%$y=R(z)$, where $R$ is rational, $\psi=T(\ta)+\sum\limits_{i}\beta^{-1}_i\log(\ta-b_i)$. We also assume that $R(z)$ is chosen in such a way that $R(z)-b_i$ has simple zeros and poles and are distinct for different $i$. 
%	$\psi(y(z))$ has only simple poles as a function of $z$.
Let $x$ be such that $dx$ is meromorphic, and 
\begin{align}\label{eq:yCaseI}
	y&=R(z),\\
	\psi&=T(\ta)+\sum\limits_{i}\beta^{-1}_i\log(\ta-b_i),
\end{align}
with the following assumptions:
\begin{itemize}
	\item  $R$ and $T$ are rational functions, %$\psi=T(\ta)+\sum\limits_{i}\beta^{-1}_i\log(\ta-b_i)$. 
	and $\beta_i$ and $b_i$ are some constants, such that $y(z)$ and $\psi(\theta)$ are not constant functions;
	\item  for all $i$ expressions $R(z)-b_i$ have only simple zeros;
	\item all $b_i$ are distinct;
	\item 	$\psi(\ta)$ is regular at the critical values of $y$;
	\item  all of the poles of $d\psi(y(z))$ which are not simultaneously poles of $dy(z)$ are distinct from the poles of $dx(z)$;
	\item $T(\theta)$ is regular at the points $\theta=b_i$.
\end{itemize}
%where $R$ and $T$ are rational functions, %$\psi=T(\ta)+\sum\limits_{i}\beta^{-1}_i\log(\ta-b_i)$. 
%and $\beta_i$ and $b_i$ are some constants, such that $y(z)$ and $\psi(\theta)$ are not constant functions. % such that all $b_i$ are distinct. We also assume that %$R(z)$ is chosen in such a way that 
%We also assume that for all $i$ expressions $R(z)-b_i$ have simple zeros which are distinct over all $i$. We assume that $\psi(\ta)$ is regular at the critical values of $y$. Furthermore, we
%assume that all of the poles of $d\psi(y(z))$ other than the poles of $y(z)$ are distinct from the poles of $dx$, and that $T(z)$ does not have poles at the points $b_i$.
% assume that all zeros of $dx$, $dy$, and $dx^\dagger$ are simple, and that the zeros of $dx$ are disjoint from the zeros of $dy$ (in that case the zeros of $dx^\dag$ are also disjoint from the zeros of $dy$).
	
	\item \label{en:caseII} %Case II:  
%	$y=R(z)+\sum\limits_{i}\alpha_i\log(y-a_i)$, where $R$ is rational% and not a polynomial of degree $\leq 1$ and there is at least one logarithmic term
%	, $\psi=\varkappa_1 \ta + \varkappa_0$.
Let $x$ be such that $dx$ is meromorphic, and 
\begin{align} \label{eq:yCaseII}
	y&=R(z)+\sum\limits_{i}\alpha_i^{-1}\log(z-a_i),\\
	\psi&=\varkappa_1 \ta + \varkappa_0,
\end{align}
with the following assumptions:
\begin{itemize}
	\item  $R$ is a rational function %$\psi=T(\ta)+\sum\limits_{i}\beta^{-1}_i\log(\ta-b_i)$. 
	and $\alpha_i$, $a_i$, and $\varkappa_i$ are some constants, such that $y(z)$ and $\psi(\theta)$ are not constant functions;
	\item all $a_i$ are distinct;
	\item $R(z)$ is regular at the points $z = a_i$.
\end{itemize}
%where $R$ is a rational function %$\psi=T(\ta)+\sum\limits_{i}\beta^{-1}_i\log(\ta-b_i)$. 
%and $\alpha_i$, $a_i$, and $\varkappa_i$ are some constants, such that $y(z)$ and $\psi(\theta)$ are not constant functions and $z=\infty$ is not a critical point of $y$. % such that all $b_i$ are distinct. We also assume that %$R(z)$ is chosen in such a way that 
%Furthermore, we assume that all zeros of $dx$, $dy$, and $dx^\dagger$ are simple, and that the zeros of $dx$ are disjoint from the zeros of $dy$ (in that case the zeros of $dx^\dag$ are also disjoint from the zeros of $dy$).
	\item \label{en:caseIII}	%Case III: 
%	$y=\sum\limits_{i}\alpha_i\log(y-a_i)$, $\psi=T(e^{\gamma_1 \ta},\dots,e^{\gamma_K \ta})+\varkappa_1 \ta + \varkappa_0$, where $T$ is a %rational function in $K$ variables and $\alpha_i$ and $\gamma_i$ are such that $\forall i,j\;\, (\alpha_i\gamma_j)\in\mathbb{Z}$. 
Let $x$ be such that $dx$ is meromorphic, and 
\begin{align}\label{eq:yCaseIII}
	y&=\sum\limits_{i}\alpha_i^{-1}\log(z-a_i),\\
	\psi&=T(e^{\gamma_1 \ta},\dots,e^{\gamma_K \ta})+\varkappa\, \ta,
\end{align}
with the following assumptions:
\begin{itemize}
\item  $T$ is a rational function in $K$ variables and $\varkappa$, $\alpha_i$, $a_i$, and $\gamma_i$ are some constants such that $y(z)$ and $\psi(\theta)$ are not constant functions;
\item all $a_i$ are distinct;
\item $\psi(\ta)$ is regular at the critical values of $y$;
\item $\forall i,j\;\, {\dfrac{\gamma_i}{\alpha_j}\in\mathbb{Z}}$.
\end{itemize}
%where $T$ is a rational function in $K$ variables and $\varkappa$, $\alpha_i$, and $\gamma_i$ are some constants such that $\forall i,j\;\, {\dfrac{\gamma_i}{\alpha_j}\in\mathbb{Z}}$, and $y(z)$ and $\psi(\theta)$ are not constant functions. Furthermore, we assume that $\psi$ is regular at the critical values of $y$, that all zeros of $dx$, $dy$, and $dx^\dagger$ are simple, and that the zeros of $dx$ are disjoint from the zeros of $dy$ (in that case the zeros of $dx^\dag$ are also disjoint from the zeros of $dy$).
%\end{proposition}
\end{enumerate}

%\begin{remark}
%	We treat the general case of multivalued $d\psi(y(z))$ elsewhere.
%\end{remark}

\begin{remark}
%	Let us note that we are unable to write the expression for $\tilde{\psi}$ in the case of general $R(z)$  in Case I, but there is a workaround. 
	
%	Let $w=\psi(y(z))$...

The restriction on zeros of $R-b_i$ in \ref{en:caseI} comes from the fact that if we do not impose it, then the procedure for determining $\tilde{\psi}$ described above gives us $\tilde{\psi}(y(z))$ but then we cannot represent it as a composition of $\tilde{\psi}(\theta)$ and $y(z)$ for some good $\tilde{\psi}(\theta)$. 
\end{remark}
\begin{remark}
	\ref{en:caseIII} can also be extended by adding logarithms of rational functions of $e^{\gamma_i\theta}$ to $\psi$, but this would require a nontrivial $\tilde{\psi}$ (i.e. not equal to $\psi$, as in \ref{en:caseI}, see below). This extension is not complicated, but we do not write it out explicitly in the present paper.
\end{remark}

\subsection{Case I}%: $y$ is meromorphic, $\psi'(y)$ is rational}
%\begin{proposition}
%	For Case I of Proposition~\ref{prop:ypsiCases} $\tilde\psi=T(y)+\sum_{i}\beta_i\frac{1}{\cS(\hbar\partial_y)}\log(y-b_i)$.
%\end{proposition}
Following the recipe described in Section~\ref{sec:tildepsi}, for~\ref{en:caseI} we obtain % from \cite{ABDKS4} we obtain that in this case 
\begin{equation}\label{eq:hatpsicaseI}
	\tilde\psi=T(\ta)+\sum_{i}\frac{1}{\beta_i\cS(\beta_i\hbar\partial_\ta)}\log(\ta-b_i).
\end{equation}
Indeed, under our assumptions for~\ref{en:caseI} for $\tilde\psi$ given by \eqref{eq:hatpsicaseI} $\forall g\geq 1$ the differential $[\hbar^{2g}]\tilde \psi(y(z),\hbar)dy(z)$  only has poles at those poles of $d\psi(y(z))$ which are not poles of $dy(z)$, and the principal parts at these points agree with~\eqref{eq:tildepsipp}.

\begin{theorem}\label{th:caseI}
	%Assume that $dx$ is meromorphic, $y$ and $\psi$ as in the Case I, and assume that zeros of $dx$ and $dy$ are disjoint, and also zeros of $dx^\dag$ and $dy$ are disjoint, where $x^\dag=x+\psi(y)$.
	Let $(\Sigma,\,x,\,y,\,\psi)$ be as in~\ref{en:caseI}
	 Assume that the four spectral curves $(\Sigma,x,y)$, $(\Sigma,y,x)$, $(\Sigma,y,x^\dag)$, $(\Sigma,x^\dag,y)$ satisfy the necessary TR conditions of Definition~\ref{def:necTR}.
	Let $\tilde \psi$ be given by \eqref{eq:hatpsicaseI}.  Then the system of differentials $\{\omega^{(g)}_n\}$ solves TR for the spectral curve $(\Sigma,x,y)$ if and only if its symplectic dual system of differentials $\{\omega^{\dag,(g)}_n\}$ solves TR for the spectral curve $(\Sigma,x^\dag,y)$.
\end{theorem}
\begin{proof}
	We follow the lines of the proof of Theorem~\ref{thm:TRsymplectic}, and the notation follows the notation of Theorem~\ref{thm:sympl-from-xy} (in particular, see~\eqref{eq:checkdef} and \eqref{eq:dagcheckdef}).
	
	%Indeed, the first and the third transformations of Theorem~\ref{thm:sympl-from-xy} preserve relations of topological recursion by Theorem \ref{Th:x-y} . 
	
	%Note that the additives from equation \eqref{eq:omegav-omegadv}, due to the form of $\tilde{\psi}$ from \eqref{hatpsicase1} are proportional to $\frac1{(y-b_i)^2}$, so linear and quadratic loop equations for differentials after the second transformations of  Theorem~\ref{thm:sympl-from-xy}  are holds by the same arguments as in the proof of Theorem \ref{thm:TRsymplectic}.
	
%The logarithmic projection property also holds because of the form of the additives  from equation \eqref{eq:omegav-omegadv}.

Transformation~\ref{en:trxy} of Theorem~\ref{thm:sympl-from-xy} preserves LogTR by Theorem~\ref{Th:x-y} (part~\ref{en:xydualTR}).
	Let us see that after the application of transformation~\ref{en:trShift1} of Theorem~\ref{thm:sympl-from-xy} LogTR still holds.
	
	%Keeping in mind that $\tilde{\psi}=\psi$, formula \eqref{eq:omegav-omegadv} says that $\omega^{\dagger,\vee,(g)}_n$ differs from $\omega^{\vee,(g)}_n$ only in $(g,n)=(0,1)$ term.
	
	%For the second transformation of Theorem~\ref{thm:sympl-from-xy} 
	Indeed, the linear loop equations for $(g,n)=(0,1)$ still hold due the fact that $\psi(\ta)$ is regular at the critical values of $x^\vee=x^{\dagger\vee}$ and thus $\psi(x^\vee(z))dx^\vee$ is regular and vanishing at the critical points of $x^\vee=x^{\dagger\vee}$. For $g>1$ we have $\omega^{\dag\vee,(g)}_1-\omega^{\vee,(g)}_1=\Bigl([\hbar^{2g}]\tilde\psi(x^\vee(z),\hbar)\Bigr)\,dx^\vee,$ where $[\hbar^{2g}]\tilde\psi(x^\vee(z),\hbar)\sim (x^\vee(z) - b_i)^{-2g}$ which is regular at the critical points of $x^\vee=x^{\dagger\vee}$ since $b_i$ do not coincide with critical values of $x^\vee=x^{\dagger\vee}$; thus linear loop equations still hold.
	Regarding the quadratic loop equations, the same reasoning as in Theorem~\ref{thm:TRsymplectic} works, as $\Bigl([\hbar^{2g}]\tilde\psi(x^\vee(z),\hbar)\Bigr)\,dx^\vee$ is symmetric under the deck transformations and thus gets killed by $\Delta$.
	
	%Clearly additives form equation \eqref{eq:omegav-omegadv} do not give new poles so the projection property (Definition \eqref{def:PP}) holds.
	
	The projection property is not completely trivial now as not only the $\omega^{(0)}_1$ differential gets changed, but $\omega^{(g)}_1$ for $g>0$ as well. %imposed only upon the $2g-2+n>0$ differentials and thus it is preserved as well, as only $\omega^{(0)}_1$ gets changed.
	The only poles that $\omega^{\dag\vee,(g)}_1-\omega^{\vee,(g)}_1=\Bigl([\hbar^{2g}]\tilde\psi(x^\vee(z),\hbar)\Bigr)\,dx^\vee$ has are at the zeros of the expressions $R(z)-b_i$, and \eqref{eq:hatpsicaseI} guarantees that the principal parts there agree with \eqref{eq:PrincipalPartsLog}, where $\alpha_i^{-1}=\beta_i^{-1}$ is the residue of $\psi(x^\vee(z))$ at the given pole, due to our assumptions. % the assumption that the zeros of $R(z)-b_i$ are simple and distinct. 
	Thus the logarithmic projection property holds.
	
	Thus, after transformation~\ref{en:trShift1} LogTR still holds, and after transformation~\ref{en:trxyback} it holds as well again due to Theorem~\ref{Th:x-y} (part~\ref{en:xydualTR}), where we again note that the zeros of $dx^\dag$ are disjoint from the zeros of $dy$, same as in Theorem~\ref{thm:sympl-from-xy}.
\end{proof}

\begin{corollary}[Extended Family~I of \cite{BDKS2}] \label{cor:FamI}
In the case of $x^\dagger(z)=\log z$  we obtain a new proof of topological recursion for the (extended version of) Family I of Orlov--Scherbin hypergeometric tau functions from \cite{BDKS2}. The Orlov--Scherbin data of~\eqref{eq:OSpartfun} and the respective spectral curve data (on which the TR is run) for this family are given in the first line of Table~\ref{tab:families}.
\end{corollary}
\begin{proof}
	Follows from Theorem~\ref{th:caseI} and Proposition~\ref{prop:OSsympl}, and the fact that the $n$-point differentials produced by the logarithmic topological recursion on the spectral curve $(\mathbb{C}\mathrm{P}^1,\log z, y(z))$ with  $y(z)$ as in~\eqref{eq:yCaseI} are given exactly by the formula~\eqref{eq:OSDualDiff} with $\hat y(z,\hbar)=R(z)$.
		
	%that  produces the following $n$-point differentials :
	%\begin{equation}
	%	\omega^{\dag,(g)}_n=\delta_{n,1}[\hbar^{2g}]\hat y_1dx_1+\delta_{(g,n),(0,2)}\tfrac{dz_1dz_2}{(z_1-z_2)^2},
	%\end{equation}
	%where $\hat y(z,\hbar)=R(z)$. Note that 
	%
	%. 
	Family~I of \cite{BDKS2} itself corresponds to setting $\beta_i=\pm 1$ and requiring that $T$ is a polynomial, and that $y(0)=0$ and $\psi(0)=0$.
\end{proof}

\subsection{Case II}
%We immediately obtain that $\tilde\psi=\psi$.

\begin{theorem}\label{th:caseII}
	%Assume that $dx$ is meromorphic, $y$ and $\psi$ as in the Case I, and assume that zeros of $dx$ and $dy$ are disjoint, and also zeros of $dx^\dag$ and $dy$ are disjoint, where $x^\dag=x+\psi(y)$.
	Let $(\Sigma,\,x,\,y,\,\psi)$ be as in~\ref{en:caseII}
	Assume that the four spectral curves $(\Sigma,x,y)$, $(\Sigma,y,x)$, $(\Sigma,y,x^\dag)$, $(\Sigma,x^\dag,y)$ satisfy the necessary TR conditions of Definition~\ref{def:necTR}.
	Let $\tilde \psi=\psi$. % be given by \eqref{hatpsicase1}. 
	Then the system of differentials $\{\omega^{(g)}_n\}$ solves LogTR for the spectral curve $(x,y)$ if and only if its symplectic dual system of differentials $\{\omega^{\dag,(g)}_n\}$ solves LogTR for the spectral curve $(x^\dag,y)$.
\end{theorem}
%\begin{proposition}[Extended Family II of \cite{BDKS2}]
%		Assume that $dx$ is meromorphic, $y$ and $\psi$ as in the Case II, and assume that zeros of $dx$ and $dy$ are disjoint, and also zeros of $dx^\dag$ and $dy$ are disjoint, where $x^\dag=x+\psi(y)$. Let $\tilde \psi=\psi$. Then the system of differentials $\{\omega^{(g)}_n\}$ solve TR for the spectral curve $(x,y)$ if and only if its symplectic dual system of differentials $\{\omega^{\dag,(g)}_n\}$ solve TR for the spectral curve $(x^\dag,y)$.
%\end{proposition}
\begin{proof}
Note that all new logarithmic singularities in $y^{\dagger\vee}$ in this case are not vital (they coincide with the singularities of $x^{\dagger\vee}$), so the proof is completely analogous to the proof of Theorem \ref{thm:TRsymplectic}.
\end{proof}

\begin{corollary}[Extended Family II of \cite{BDKS2}] \label{cor:FamII}
	In the case of $x^\dagger(z)=\log z$  we obtain a new proof of topological recursion for the (extended version of) Family II of Orlov--Scherbin hypergeometric tau functions from \cite{BDKS2}. The Orlov--Scherbin data of~\eqref{eq:OSpartfun} and the respective spectral curve data (on which the TR is run) for this family are given in the second line of Table~\ref{tab:families}.
\end{corollary}
\begin{proof}
Follows from Theorem~\ref{th:caseII} and Proposition~\ref{prop:OSsympl}, and the fact that the $n$-point differentials produced by the logarithmic topological recursion on the spectral curve $(\mathbb{C}\mathrm{P}^1,\log z, y(z))$ with  $y(z)$ as in~\eqref{eq:yCaseII} are given exactly by the formula~\eqref{eq:OSDualDiff} with \begin{equation}
	\hat y(z,\hbar) = R(z)+
\sum\limits_{i}\dfrac{1}{\alpha_i\cS(\alpha_i\hbar z\partial_z)}\log(z-a_i).
\end{equation} 

Family~II of \cite{BDKS2} itself corresponds to setting $\alpha_i=\pm 1$, $\varkappa_0=0$, and requiring that $y(0)=0$ and that $y$ does not have a vital singularity at $\infty$ with respect to $x$.%^\dagger$.
\end{proof}

The topological recursion for an extension of Family~II of \cite{BDKS2} was also earlier obtained in~\cite{kramer2022kp}.

\subsection{Case III}%: $e^y$ is meromorphic, $\psi(y)$ is a rational function in $e^y$}
%We immediately obtain that $\tilde\psi=\psi$.

\begin{theorem}\label{th:caseIII}
	%Assume that $dx$ is meromorphic, $y$ and $\psi$ as in the Case I, and assume that zeros of $dx$ and $dy$ are disjoint, and also zeros of $dx^\dag$ and $dy$ are disjoint, where $x^\dag=x+\psi(y)$.
	Let $(\Sigma,\,x,\,y,\,\psi)$ be as in~\ref{en:caseIII}
	Assume that the four spectral curves $(\Sigma,x,y)$, $(\Sigma,y,x)$, $(\Sigma,y,x^\dag)$, $(\Sigma,x^\dag,y)$ satisfy the necessary TR conditions of Definition~\ref{def:necTR}.
	Let $\tilde \psi=\psi$. % be given by \eqref{hatpsicase1}. 
	Then the system of differentials $\{\omega^{(g)}_n\}$ solves LogTR for the spectral curve $(x,y)$ if and only if its symplectic dual system of differentials $\{\omega^{\dag,(g)}_n\}$ solves LogTR for the spectral curve $(x^\dag,y)$.
\end{theorem}

%\begin{proposition} [Extended Family III of \cite{ABDKS4, Section 3.11}]
%			Assume that $dx$ is meromorphic, $y$ and $\psi$ as in the Case III, and assume that zeros of $dx$ and $dy$ are disjoint, and also zeros of $dx^\dag$ and $dy$ are disjoint, where $x^\dag=x+\psi(y)$. Let $\tilde \psi=\psi$. Then the system of differentials $\{\omega^{(g)}_n\}$ solve TR for the spectral curve $(x,y)$ if and only if its symplectic dual system of differentials $\{\omega^{\dag,(g)}_n\}$ solve TR for the spectral curve $(x^\dag,y)$.
%\end{proposition}
\begin{proof}
Note that all new logarithmic singularities in $y^{\dagger\vee}$ in this case are not vital (they coincide with the singularities of $x^{\dagger\vee}$), so the proof is completely analogous to the proof of Theorem \ref{thm:TRsymplectic}.
\end{proof}

\begin{corollary}[Extended Family III of {\cite[Section 3.11]{ABDKS4}}] \label{cor:FamIII}
	In the case of $x^\dagger(z)=\log z$  we obtain a proof of (logarithmic) topological recursion for the (extended version of) Family III of Orlov--Scherbin hypergeometric tau functions from {\cite[Section 3.11]{ABDKS4}}. The Orlov--Scherbin data of~\eqref{eq:OSpartfun} and the respective spectral curve data (on which the TR is run) for this family are given in the third line of Table~\ref{tab:families}.
\end{corollary}
\begin{proof}
	Follows from Theorem~\ref{th:caseIII} and Proposition~\ref{prop:OSsympl}, and the fact that the $n$-point differentials produced by the logarithmic topological recursion on the spectral curve $(\mathbb{C}\mathrm{P}^1,\log z, y(z))$ with  $y(z)$ as in~\eqref{eq:yCaseIII} are given exactly by the formula~\eqref{eq:OSDualDiff} with \begin{equation}
		\hat y(z,\hbar) = \sum\limits_{i}\dfrac{1}{\alpha_i\cS(\alpha_i\hbar z\partial_z)}\log(z-a_i).
		\end{equation}
	
	 Family~III of~\cite[Section 3.11]{ABDKS4} itself corresponds to setting $\alpha_i=\pm 1$, $\varkappa=0$, $K=1$, $T(\zeta)=\alpha(\zeta-1)$, $\gamma_1=1$, and requiring that $y(0)=0$.
	\end{proof}

%A different proof of topological recursion in this case is given in~\cite{ABDKS-FamIII}.
%%%%%%%%%%
%%%%%%%%%%
%%%%%%%%%%
%%%%%%%%%%
%%%%%%%%%%
\subsection{Summary of implications for weighted Hurwitz numbers}\label{sec:HurwSummary}

Let us discuss the statements of Corollaries~\ref{cor:FamI}, \ref{cor:FamII}, and \ref{cor:FamIII} in a bit more detail. 

Recall formula~\eqref{eq:OSpartfun} for the Orlov-Scherbin partition function:
\begin{equation}
	Z^{\mathrm{OS}}=\sum_{\lambda}\shin_\lambda(t)\shin_\lambda(\hbar^{-1} s)\exp\left(\sum_{(i,j)\in \lambda}\hat\psi(\hbar(i-j),\hbar)\right),
\end{equation}
which defines weighted Hurwitz numbers $h^{(g)}_{k_1,\dots,k_n}$ by
\begin{equation}\label{eq:HurwNums}
	h^{(g)}_{k_1,\dots,k_n}\coloneqq % \hbar^{-n} 
	[\hbar^{2g-2+n}]\,\frac{\partial^n \log Z^{\mathrm{OS}}}{\prod_{i=1}^n \partial t_{k_i}}\Bigg|_{t_1=t_2=\cdots=0},	
\end{equation}
see~\cite[Section~1.2]{BDKS2} for more details. Note (or recall) that this partition function and thus the respective weighted Hurwitz numbers depend on some formal series $\hat \psi(\theta,\hbar)$ and some series $\hat y(\theta,\hbar)=\sum_{j=1}^\infty s_j(\hbar)z^j$ which encodes the $s$-parameters. We call the pair $(\hat \psi, \hat y)$ the \emph{Orlov--Scherbin data}.

In general the numbers~\eqref{eq:HurwNums} have an explicit combinatorial meaning via the works of Guay-Paquet and Harnad \cite{guay2017generating,harnad2016weighted} (they do not consider the case when $\hat \psi$ and $\hat y$ depend on $\hbar$ themselves, but it can easily be included into their considerations) in terms of weighted counts of paths in the Cayley graph of the symmetric group. For certain particular cases of $\hat \psi$ and $\hat y$ these numbers are certain Hurwitz-type numbers well known from the literature, see Tables~\ref{tab:Hurw},~\ref{tab:Hurw2}.

%where $\hat{\psi}(\theta,\hbar)$ and $\hat y(z,\hbar)$ are some series in $\hbar$ and $\theta$ and $z$, respectively, and the parameters $s_i=s_i(\hbar)$ for the second Schur function are the expansion coefficients defined by $\hat y(z,\hbar)=\sum_{j=1}^\infty s_j(\hbar)z^j$. This can be used as a definition 

%$	x(z)  = \log z - \psi(y(z))$
Corollaries~\ref{cor:FamI}, \ref{cor:FamII}, and \ref{cor:FamIII} state that for three very general families corresponding to the Orlov--Scherbin data given in Table~\ref{tab:families}, the expansion coefficients of the $n$-point differentials produced by the logarithmic topological recursion (which coincides with the original TR for Families I and II, but not for Family III) on the respective spectral curves give the respective weighted Hurwitz numbers. Specifically, near $X=0$, 
\begin{equation}
	\omega^{(g)}_n =\sum_{k_1,\dots,k_n=1}^\infty h^{(g)}_{k_1,\dots,k_n} \prod_{i=1}^n k_i X_i^{k_i-1} dX_i + \delta_{(g,n),(0,2)} \dfrac{dX_1dX_2}{(X_1-X_2)^2},
\end{equation}
where $X_i=X(z_i)$ and $X(z)=e^{x(z)}$.

\begin{table}[H]
	\scalebox{0.82}{
		\begin{tabular}{|c||l|l|}
			%\begin{TAB}(@,0.5cm,2cm)[5pt]{|c|c|c|c|c|}{|c|c|c|}
			\hline
			\Gape[0.2cm]{Family} & Orlov--Scherbin data & Spectral curve data\\ %$\psi(\ta)$ & $y(z)$ & $\hat\psi(\ta,\hbar)$ & $\hat y(z,\hbar)$\\
			\hline
			\hline
			\Gape[0.5cm]{I} &  $\hat\psi(\ta,\hbar)=\cS(\hbar \partial_\ta)T(\ta)+\sum\limits_{i}\dfrac{\cS(\hbar \partial_\ta)}{\beta_i\cS(\beta_i\hbar\partial_\ta)}\log(\ta-b_i)$ %$\cS(\hbar \partial_\ta)\dfrac{P_1(y)}{P_2(y)}+\log\left(\dfrac{P_3(y)}{P_4(y)}\right)$
			&  $x(z)=\log z - \left(T(y(z))+\sum\limits_{i}\beta^{-1}_i\log(y(z)-b_i)\right)$  %$\dfrac{P_1(y)}{P_2(y)}+ \log\left(\dfrac{P_3(y)}{P_4(y)}\right)$
			 \\
			%\hline
			\phantom{\Gape[0.5cm]{I}}& $\hat y(z,\hbar)=R(z)$& $%\dfrac{R_1(z)}{R_2(z)}
			y(z)=R(z)$  \\ 
			%\hline
			\hline			
			\Gape[0.5cm]{II} & $\hat\psi(\ta,\hbar) = \varkappa_1 \ta + \varkappa_0$ &  $x(z)=\log z - \left(\varkappa_1 y(z) + \varkappa_0\right)$\\
			%\hline			
			\phantom{\Gape[0.5cm]{I}}& $%Q_1(z)+
			\hat y(z,\hbar) = R(z)+
			\sum\limits_{i}\dfrac{1}{\alpha_i\cS(\alpha_i\hbar z\partial_z)}\log(z-a_i)$  & $%Q_1(z)+
			y(z)=R(z)+\sum\limits_{i}\alpha_i^{-1}\log(z-a_i)$  \\
			%\hline
			\hline
			\Gape[0.5cm]{III}& $\hat\psi(\ta,\hbar) = \cS(\hbar \partial_\ta)T(e^{\gamma_1 \ta},\dots,e^{\gamma_K \ta})+\varkappa\, \ta$ &  $x(z) = \log z - \left(T(e^{\gamma_1 y(z)},\dots,e^{\gamma_K y(z)})+\varkappa\, y(z)\right)$  \\
			%\hline			
			\phantom{\Gape[0.5cm]{I}} & $%Q_1(z)+
			\hat y(z,\hbar) = \sum\limits_{i}\dfrac{1}{\alpha_i\cS(\alpha_i\hbar z\partial_z)}\log(z-a_i)$& $%Q_1(z)+
			y(z) = \sum\limits_{i}\alpha_i^{-1}\log(z-a_i)$ \\
			\hline
	\end{tabular} }
	%	\renewcommand{\arraystretch}{1}
	%	\vspace{0.1cm}
	\caption{%Cases of weighted Hurwitz numbers data $\psi$, $y$ where we can prove the projection property and know the unique $\hbar^2$-deformation ($\hat{\psi}$, $\hat y$). 
		Families of weighted Hurwitz numbers, represented by the respective Orlov--Scherbin data, for which symplectic duality implies the (logarithmic) topological recursion on the respective spectral curve $\left(\mathbb{C}\mathrm{P}^1,x(z),y(z)\right)$ via Corollaries~\ref{cor:FamI}, \ref{cor:FamII}, and \ref{cor:FamIII} respectively.		
		Here $T$ and $R$ are some rational functions and $a_i$, $\alpha_i$, $b_i$, $\beta_i$, $\varkappa_i$, $\varkappa$, $\gamma_i$ are some constants such that the conditions given in the respective Cases of Section~\ref{sec:principleAtWork} are satisfied.}
		 % some polynomials (of degree $\geq 0$)
		%such that the natural analytic assumptions of Definition~\ref{def:naa} are satisfied.}
	%, with the requirement that $\psi(y)$ and $y(z)$ are both nonzero but vanishing at zero, %, with $\psi'(y)$ and $y'(z)$ possibly having only simple poles;
	%	and $\alpha\neq 0$ is a number.}%; additionally, for Family I either $\deg R_2 = 0$ or $\deg P_1 >0$. }
	\label{tab:families}
\end{table}

Family I %a
includes, as special cases, weighted Hurwitz problems listed in Table \ref{tab:Hurw}.

\begin{table}[H]
\renewcommand{\arraystretch}{1.5}
\scalebox{0.9}{
\begin{tabular}{|c|c|c|}
	\hline
	Hurwitz numbers & $\hat\psi(\ta,\hbar)$ & $x(z)$\\
	\hline
	\hline
	simple & $\ta$ & $\log z - y(z)$ \\
	\hline
	%$r$-spin
	r-spin
	& $\cS(\hbar \partial_\ta)\, \ta^r$ & $\log z - y^r(z)$ \\
	\hline
	monotone & $\log\big(1/(1-\ta)\big)$ & $\log z - \log\big(1/(1-y(z))\big)$ \\
	\hline
	strictly monotone & $\log(1+\ta)$ & $\log z - \log(1+y(z)$ \\
	\hline
	hypermaps & $\log\big((1+u\,\ta)(1+v\,\ta)\big)$ & $\log z - \log\big((1+u\,y(z))(1+v\,y(z))\big)$\\
	\hline
	%BMS numbers & $\log\big((1+\ta)^m\big)$ & $\psi(\ta)$ \\
	polynomially weighted& $\log\left(\sum_{k=1}^d c_k \ta^k\right)$ & $\log z - \log\left(\sum_{k=1}^d c_k y^k(z)\right)$\\	
	\hline
	rationally weighted& $\log\left(\frac{\sum_{k=1}^{d_1} a_k \ta^k}{\sum_{k=1}^{d_2} b_k \ta^k}\right)$ & $\log z - \log\left(\frac{\sum_{k=1}^{d_1} a_k \ta^k}{\sum_{k=1}^{d_2} b_k \ta^k}\right)$\\
	%	\hline
	%	general weighted & $\exp\left({\sum_{k=1}^\infty c_k y^k}\right)$ & $\psi(y)$\\
	\hline
	%weighted & 		
	\end{tabular} } 
\scalebox{0.9}{ \begin{tabular}{|c|c|c|}
	\hline
	Variations &  $\hat y(z,\hbar)$ & $y(z)$ \\
	\hline
	\hline
	single & $z$ & $z$ \\
	orbifold & $z^q$ & $z^q$ \\
	\Gape[0.1cm]{\shortstack{polynomial\\ double}} & \raisebox{0.2cm}{$\sum_{k=1}^d s_kz^k$} & \raisebox{0.2cm}{$\sum_{k=1}^d s_kz^k$} \\
	%	\hline
	%	double & $\sum_{k=1}^\infty s_kz^k$ & $y(z)$ \\
	\hline
	%weighted & 		
	\end{tabular} }
	\renewcommand{\arraystretch}{1}
	%	\vspace{0.1cm}
	\caption{Types of Hurwitz numbers (known from the literature) belonging to Family I, with their Orlov--Scherbin data and the respective spectral curve data.}
	\label{tab:Hurw}
\end{table}

Family II, out of already known examples from the literature, includes the cases of the extended Ooguri-Vafa partition function for the HOMFLY-PT polynomials of torus knots~\cite{HOMFLY2,HOMFLY1}, the Mari\~no-Vafa numbers~\cite{kramer2022kp} (which correspond to triple Hodge integrals via the Mari\~no--Vafa formula~\cite{MV02,liu2003proof}), and of the simple double Hurwitz numbers (the latter case is also included in Family I), see Table \ref{tab:Hurw2}.

\begin{table}[H]
\renewcommand{\arraystretch}{1}
\scalebox{0.9}{
\begin{tabular}{|c|c|c|c|c|}
\hline
\Gape[0.3cm]{Hurwitz numbers} & $\hat\psi(\ta,\hbar)$ & $\hat y(z,\hbar)$ & $x(z)$ &$y(z)$\\
\hline
\hline
\Gape[0.3cm]{simple polynomial double} & $\ta$ & $\sum_{k=1}^d s_kz^k$ & $\log z - y(z)$ & $\sum_{k=1}^d s_kz^k$ \\
\hline
\Gape[0.3cm]{extended Ooguri--Vafa} & $\frac{P}{Q}\,\ta$ &  $\frac{1}{\cS(\hbar z\partial_z)}\log \left(\frac{1-A^{-1}z}{1-Az}\right)$ & $\log z - \frac{P}{Q}\,y(z)$ & $\log \left(\frac{1-A^{-1}z}{1-Az}\right)$ \\
\hline
\Gape[0.3cm]{Mari\~no--Vafa} & $-\ta/w$ &  $-\frac{1}{\cS(\hbar z\partial_z)}\log \left(1-\beta w z\right)$ & $\log z +y(z)/w$ & $-\log \left(1-\beta w z\right)$ \\
\hline
\end{tabular} }
\renewcommand{\arraystretch}{1}
%	\vspace{0.1cm}
\caption{Types of Hurwitz numbers (known from the literature) belonging to Family II, with their Orlov--Scherbin data and the respective spectral curve data.} %, with their $(\psi,y)$-data and the unique $\hbar^2$-extension $(\hat\psi, \hat y)$.}
\label{tab:Hurw2}
\end{table}

\begin{remark}
	Note that proving TR for various Hurwitz-type numbers (listed in Tables~\ref{tab:Hurw},~\ref{tab:Hurw2}) used to be a very complicated task, with numerous papers dedicated to careful case-by-case analysis which allowed to do it (e.g. \cite{DKOSS-ELSV,DLPS,HOMFLY1,DKPS-rspin,BDKLM}). Even the paper~\cite{BDKS2}, in which TR was proved for two rather general families of weighted Hurwitz numbers (which include all the above cases), required some careful pole structure analysis. But in the present paper we get a considerably more general result (which, in particular, includes all of the above) basically for free, from the known properties of symplectic duality and from the observation of Proposition~\ref{prop:OSsympl} that the explicit formulas for the generating functions of weighted Hurwitz numbers obtained in~\cite{BDKS1} coincide with the result of applying symplectic duality to a certain simple system of differentials.
\end{remark}

\begin{remark}
	Note that the original topological recursion and the logarithmic topological recursion require $dx$ to have simple critical points and $dy$ to be regular and nonvanishing at the zero locus of $dx$. This is not the case for some particular values of parameters for some of the entries of Tables~\ref{tab:Hurw} and~\ref{tab:Hurw2} (but when the parameters are in general position these requirements do hold). However, there are ways to define generalizations of topological recursion which do not have such requirements. It is the Bouchard-Eynard recursion~\cite{BE13} and the generalized topological recursion (GenTR)~\cite{ABDKS7}. The respective statements for the cases when the requirements of the original TR and the LogTR do not hold can be obtained by taking limits. This is quite non-trivial for the Bouchard-Eynard recursion (see~\cite{limits}), but is much easier for GenTR (this is discussed in~~\cite{ABDKS7}).
\end{remark}

\begin{remark}
	Note that when we talk about ``double Hurwitz numbers'' in the present text, we treat the $s$-variables in~\eqref{eq:OSpartfun} as parameters which get some values prescribed by $\hat y$, and not as formal variables like the $t$-variables for the partition function to be expanded in. This means that, depending on the choice of $\hat y$, our Hurwitz numbers~\eqref{eq:HurwNums} are certain weighted sums over $L$ of the proper double Hurwitz numbers which depend on two Young diagrams $K=(k_1,\dots,k_n)$ and $L=(\ell_1,\dots,\ell_m)$ (and correspond to the expansion coefficients of the logarithm of~\eqref{eq:OSpartfun} in both sets of variables $t$ and $s$).
\end{remark}

\section{More examples} \label{sec:MoreExamples}

In this section we present several examples supplementing those in~\cite{ABDKS4} and demonstrating interaction between the notions of $x-y$ duality, symplectic duality, topological recursion, and LogTR.

\subsection{Topological recursion with unramified \texorpdfstring{$y$}{y}-function}

Proposition~3.1 of~\cite{ABDKS4} provides explicit formulas for TR differentials for the spectral curves of genus zero and $(x,y)$ with arbitrary~$x$ such that $dx$ is rational and with unramified $y$, that is, $dy$ never vanishes and $y$ can be parametrized either as $y=z$ or as $y=\log z$ with a global coordinate $z$. %Since $dx$ is rational, $x$ can be written as 
For $x$ we can write
\begin{equation}
	x(z)=\sum_{i=1}^M\frac{\log(z-a_i)}{\alpha_i}+R(z),
\end{equation}
where $a_1,\dots,a_M$ are the LogTR-vital singularities of~$x$ and $1/\alpha_i$ are the corresponding residues of~$dx$. The term $R(z)$ may contain contributions of other non-LogTR-vital logarithmic singularities, so we do not assume that $R$ is rational.

The formulas of~\cite[Proposition~3.1]{ABDKS4} express the TR differentials for such a spectral curve as $x-y$ dual to the trivial LogTR differentials for the respective spectral curve $(y,x)$. Namely, we have
\begin{align}\label{eq:y=z}
y&=z:&\omega^{(g)}_n&=[\hbar^{2g-2+n}](-1)^n\prod_{i=1}^n
\Bigl(\sum_{r=0}^\infty \bigl(-d_i\tfrac{1}{dx_i}\bigr)^{r}[u_i^r]e^{u_i(\cS(u_i\hbar\partial_{z_i})\hat x_i-x_i)}dz_i\Bigr)
\\\notag&&&
\qquad\qquad(-1)^{n-1}\!\!\sum_{\sigma\in\{n\text{-cycles}\}}\prod_{i=1}^n\frac{1}{z_i+\frac{u_{i}\hbar}{2}-z_{\sigma(i)}+\frac{u_{\sigma(i)}\hbar}{2}},
\\\label{eq:y=logz}
y&=\log z:&\omega^{(g)}_n&=[\hbar^{2g-2+n}](-1)^n\prod_{i=1}^n
\Bigl(\sum_{r=0}^\infty \bigl(-d_i\tfrac{1}{dx_i}\bigr)^{r}[u_i^r]e^{u_i(\cS(u_i\hbar z_i\partial_{z_i})\hat x_i-x_i)}dz_i\Bigr)
\\\notag&&&
\qquad\qquad(-1)^{n-1}\!\!\sum_{\sigma\in\{n\text{-cycles}\}}\prod_{i=1}^n\frac{1}{e^{\frac{u_{i}\hbar}{2}}z_i-e^{-\frac{u_{\sigma(i)}\hbar}{2}}
z_{\sigma(i)}},
\end{align}
where
\begin{equation}	
\hat{x}(z,\hbar)=\sum_{i=1}^M\frac1{\cS(\alpha_i\hbar\partial_{y_i})}\frac{\log(z-a_i)}{\alpha_i}+R(z).
\end{equation}
%where $\hat x(z)$ is a suitable $\hbar$-deformation of the function~$x(z)$ described in~\cite{ABDKS4}. 

To be precise, Eq.~\eqref{eq:y=z} and~\eqref{eq:y=logz} provide the LogTR differentials which coincide, in fact, with TR differentials in the case $y=z$ for any $x$ and also in the case $y=\log z$ if $dx$ has a pole both at $z=0$ and $z=\infty$. The next theorem shows that these formulas are contained implicitly in the closed expression for the Orlov--Scherbin family derived earlier in \cite{BDKS2} and \eqref{eq:OS}.

\begin{proposition}\label{prop:xytriv-OS}
1. If $y=z$, then the differentials~\eqref{eq:y=z} can be obtained as a special case of the Orlov--Scherbin differentials~\eqref{eq:OS} with
\begin{equation}
\hat y(z)=z,\qquad \tilde\psi(\ta)=\frac{1}{\cS(\hbar\partial_\ta)}\log \ta-\hat x(\ta).
\end{equation}

2. If $y=\log z$, then the differentials~\eqref{eq:y=logz} can be obtained as a special case of the Orlov--Scherbin  differentials~\eqref{eq:OS} with
\begin{equation}
\hat y(z)=\log z-\hat x(z),\qquad \tilde\psi(\ta)=\ta.
\end{equation}
\end{proposition}

In both cases Eq.~\eqref{eq:OS} should be applied formally ignoring the condition $y(0)=\psi(0)=0$ required for the Orlov--Scherbin family. It follows that this formula provides the correct answer even though the corresponding differentials could be singular at $z=0$ and have no VEV presentation.

\begin{proof} In the case $y=z$ we treat expressions for the differentials of the Orlov--Scherbin family as a special case of symplectic duality and apply Theorem~\ref{thm:sympl-from-xy}. We start with the trivial differentials associated with the spectral curve $(z,\log z)$. On the first step we apply $x-y$ duality and obtain LogTR differentials for the spectral curve $(\log z,z)$ that we denote by $\tilde\omega^{(g)}_n$. By Example~\ref{ex:(logz,z)} we know that they are trivial for $n>1$ and in the case $n=1$ they are given by~\eqref{eq:(logz,z)},
\begin{equation}
\tilde\omega^{(g)}_1=\Bigl([\hbar^{2g}]\tfrac{1}{\cS(\hbar\partial_z)}\log z\Bigr)\;dz.
\end{equation}
On the next step we modify these $(g,1)$ differentials by adding terms corresponding to the chosen function $\tilde\psi$ and obtain the differentials $\omega^{\vee,(g)}_n$ which are trivial for $n>1$ and in the case $n=1$ given by
\begin{equation}
\omega^{\vee,(g)}_1=\tilde\omega^{(g)}_1-[\hbar^{2g}]\tilde\psi(z)\;dz=[\hbar^{2g}]\hat x\;dz.
\end{equation}
 These are exactly the differentials $x-y$ dual to the differentials~\eqref{eq:y=z}. Therefore, when we apply $x-y$ duality to these differentials we obtain exactly the differentials~\eqref{eq:y=z}.

 In the case $y=\log z$ the argument is even simpler. We notice that in the case $\tilde\psi(\ta)=\ta$ we have
 \begin{equation}
e^{-v\psi(\ta)}\partial_\ta^r e^{v\cS(v\hbar\partial_\ta)\tilde\psi(\ta)}=v^r,
\qquad
(\cS(u_i\hbar z\partial_{z})-1)\log z=0.
\end{equation}
Therefore, the right hand side of~\eqref{eq:OS} simplifies to~\eqref{eq:y=logz} with $u_i$ replaced by $-u_i$.
\end{proof}

\subsection{Shifted Kontsevich--Witten potential}

The Kontsevich--Witten (KW) partition function can be computed by TR with the spectral curve $(x,y)=(z^2/2,z)$. We would like to relate this TR to the TR differentials of the Orlov--Scherbin family. One of the ways to do that is to apply Proposition~\ref{prop:xytriv-OS} which claims that the TR differential for the spectral curve $(z^2/2,z)$ can be computed by the closed formulas~\eqref{eq:OS} for the Orlov--Scherbin family with the following specialization of parameters:
\begin{equation}
\hat y(z)=y(z)=z,
\qquad \tilde \psi(\ta)=\frac{1}{\cS(\hbar\partial_\ta)}\log \ta-\frac{\ta^2}{2}.
\end{equation}
This specialization does not satisfy the requirement $\psi(0)=0$ and does not imply the Orlov--Scherbin type of the associated tau function. Another way to relate the KW partition function to the Orlov--Scherbin family is to modify slightly the function $y$, which corresponds to a certain shift of times in the KW partition function. For example, consider the following spectral curve data
\begin{equation}\label{eq:KWmod}
\left(\frac {z^2}{2},~\frac{z}{1-a^2z^2}\right),
\end{equation}
where $a\ne0$ is a complex parameter. The initial TR differential for this spectral curve is
\begin{equation}
\omega^{(0)}_1=\frac{z}{1-a^2z^2}\,d\Bigl(\frac{z^2}{2}\Bigr)=z\;d\bigl(-\tfrac{1}{2 a^2}\log(1-a^2 z^2)\bigr)
\end{equation}
and, since the zeros of the differentials $d(z^2/2)=zdz$ and $d\bigl(-\tfrac{1}{2 a^2}\log(1-a^2 z^2)\bigr)=zdz/(1-a^2 z^2)$ coincide, it shows that the TR for the spectral curves $(z^2/2,z/(1-a^2z^2))$ and $\bigl(-\tfrac{1}{2 a^2}\log(1-a^2 z^2),\;z\bigr)$ also coincide. This leads to the following variety of possibilities for computing these differentials, see \cite{BM,ZhouKW,Eynard-intersections,Chen,DKPS-cut-and-join,AleKPIII,KN,ABDKS4}.

\begin{proposition}
The following five ways for the computation of the TR differentials $\omega^{(g)}_n$ associated with the spectral curve~\eqref{eq:KWmod} are equivalent to one another:
\begin{enumerate}
\item apply directly the TR procedure for the spectral curve data~\eqref{eq:KWmod} in a straightforward way;

\item use the following explicit expressions for these differentials in terms of intersection numbers of $\psi$ and $\kappa$ classes over moduli spaces:
\begin{align}
\xi_k(z)&=\Bigl(-\frac{1}{z}\frac{d}{dz}\Bigr)^k\frac{1}{z}=\frac{(2k-1)!!}{z^{2k+1}},
\\\omega^{(g)}_n&=\sum_{k_1,\dots,k_n}
\int_{\overline{\mathcal M}_{g,n}}e^{\sum_{k=1}^\infty s_ka^{2k}\kappa_k}\psi_1^{k_1}\dots\psi_n^{k_n}
\prod_{i=1}^nd\xi_{k_i}\!(z_i),
%\;d\xi_{k_1}\!(z_1)\dots d\xi_{k_n}\!(z_n)
\end{align}
where the constants $s_k$ are determined by the expansion $e^{-\sum_{k=1}^\infty s_kt^k}=\sum_{k=0}^\infty (2k+1)!!t^k$;
\item apply $x-y$ duality formula~\eqref{eq:y=z} to the spectral curve 
\begin{align}
(x,y)=\bigl(-\tfrac{1}{2 a^2}\log(1-a^2 z^2),\;z\bigr)
\end{align} 
with
\begin{equation}
%x(z)=-\tfrac{1}{2 a^2}\log(1-a^2 z^2),\qquad
\hat x(z)=\frac{1}{\cS(2a^2\hbar\partial_z)}x(z);
\end{equation}
\item use an expression for these differentials in terms of triple Hodge integrals associated with the spectral curve of item~(3):
\begin{align}
\xi_k(z)&=\Bigl(-\frac{d}{dx(z)}\Bigr)^k\frac{1}{z},
\\\Lambda(q)&=\sum_{k=1}^g q^k\lambda_k,
\\\omega^{(g)}_n&=\sum_{k_1,\dots,k_n}
\int_{\overline{\mathcal M}_{g,n}}\Lambda(-2a^2)\Lambda(-2a^2)\Lambda(a^2)
e^{\sum_{k=1}^\infty \frac{(-1)^ka^{2k}}{k}\kappa_k}\psi_1^{k_1}\dots\psi_n^{k_n}
\;\prod_{i=1}^n d\xi_{k_i}\!(z_i);
%d\xi_{k_1}\!(z_1)\dots d\xi_{k_n}\!(z_n)
\end{align}
\item
apply Eq.~\eqref{eq:OS} for the Orlov--Scherbin family with the following specialization of parameters:
\begin{gather}
\hat y(\zeta)=y(\zeta)=\zeta,
\qquad
\psi(\ta)=-\log(1-a^3\ta),
\quad\tilde\psi(\ta)=\frac{1}{\cS(\hbar\partial_\ta)}\psi(\ta),
\\X=\zeta\,e^{-\psi(y(\zeta))}=\zeta-a^3\zeta^2,
\end{gather}
where $\zeta=\frac{1-a\,z}{2a^3}$ is a new affine coordinate on $\Sigma=\C \mathrm{P}^1$.
\end{enumerate}
\end{proposition}

The equivalence of all these relations is known, we just collected them together in a compact form. The parameters $(-2a^2,-2a^2,a^2)$ of triple Hodge integrals are inverse to the residues $\bigl(\frac{-1}{2a^2},\frac{-1}{2a^2},\frac{1}{a^2}\bigr)$ of the differential $dx$ of item~(3) at its three poles. The triple Hodge integrals without insertion of $\kappa$ classes correspond to the choice of $y$ such that $\frac{dy}{dx}=\frac{1}{z}$, that is, $y=\frac{1}{2a}\log\frac{1+a z}{1-a z}$. Since our choice $y=z$ is different, an extra insertion of $\kappa$ classes is needed.

Note also that the computation of item~(5) shows that the coefficients of the power expansion of~$\omega^{(g)}_n$ at the point $z_i=1/a$ in the local coordinates $\zeta_i$ are monotone Hurwitz numbers.

\subsection{Atlantes vs spin Hurwitz numbers}

Consider the differentials of the Orlov--Scherbin family with
\begin{equation}
\hat y=y=z.
\end{equation}
If $e^{\psi(\ta)}$ is rational (and such that all its zeros and poles distinct from $0$ and $\infty$ are simple) then~\eqref{eq:OS} with
\begin{equation}\label{eq:tpsiACEH}
X=\frac{z}{e^{\psi(y(z))}},\qquad \tilde\psi(\ta)=\frac{1}{\cS(\hbar\partial_\ta)}\psi(\ta)
\end{equation}
produces TR differentials for the spectral curve
\begin{equation}\label{eq:spincurve}
x=\log X=\log z-\psi(z),\quad y(z)=z,
\end{equation}
see~\cite{ACEH,BDKS2}. By Proposition~\ref{prop:xytriv-OS}, these differentials can also be given by~\eqref{eq:y=z} with
\begin{equation} \label{eq:OS-atlantes}
\hat x=\frac{1}{\cS(\hbar\partial_z)}(\log z-\psi(z)).
\end{equation}

Assume now that $\psi(\ta)$ is itself rational, for instance, $\psi(\ta)=\ta^r$. Then the differentials produced by the same formula are \emph{not} necessarily TR differentials any more: they may attain extra poles at $z=\infty$. The numbers appearing in the power expansion of these differentials at the origin in the local coordinates~$X_i$ are known as \emph{Atlantes Hurwitz numbers}~\cite{ALS}. In order to correct the situation and to obtain the true TR differentials for the same spectral curve, one should apply~\eqref{eq:OS} to a different $\hbar$-deformation of $\psi$, namely, to $\tilde\psi(\ta)=\psi(\ta)$ (we still assume that $\psi(\ta)$ is rational). Equivalently, we can apply~\eqref{eq:y=z} with
\begin{equation}\label{eq:xhatspin}
\hat x=\frac{1}{\cS(\hbar\partial_z)}\log z-\psi(z).
\end{equation}
The corresponding differentials enumerate \emph{spin Hurwitz numbers}.

\subsubsection{Transalgebraic topological recursion}

In order to reproduce the Atlantes Hurwitz numbers in the case when $\psi(\ta)=\ta^r$, it was suggested in~\cite{BKW} to consider the following family of spectral curves
\begin{align} \label{eq:BKW-setup}
	X_N(z) & = z \Big(1+ \frac{\psi(z)}{N}\Big)^{-N}; && Y_N(z) = \frac{z}{X_N(z)}; && (\omega_N)^{(0)}_1 = Y_NdX_N,
\end{align}
depending on a parameter $N$; compute the differentials $(\omega_N)^{(g)}_n$ and take the limit $N\to \infty$ (it is a special case of the proposal made in~\cite{BKW} under the name transalgebraic topological recursion). The proof that the limit procedure works in this case (and gives the desired differentials) is given in~\cite{BKW}, and it is a quite subtle analytic argument. The methods of the present paper give an alternative proof of their main theorem. 

\subsubsection{Atlantes Hurwitz numbers}  In order to apply our technique to the limit procedure suggested in~\cite{BKW}, we have to pass from $X_N$ and $Y_N$ to $x_N = \log X_N$ and $y_N=z$. If we do it naively, we lose some of the critical points of $dX_N$ (the situation is similar to the one for Bousquet-M\'elou--Schaeffer differentials discussed in~\cite{BDS,ABDKS4}). To remedy this problem, we assume that $N$ is a positive integer and introduce an $\epsilon$-deformation of the family of differentials constructed in~\cite{BKW} by
	\begin{align} \label{eq:Nepsilonspectralcurve}
		\psi_{N,\epsilon}(z)&=\sum\nolimits_{i=1}^N\,\log P^N_{i,\epsilon}(z),
		\\ 
		x_{N,\epsilon}(z)&=\log z-\psi_{N,\epsilon}(z)=\log z-\sum\nolimits_{i=1}^N\,\log P^N_{\epsilon,i}(z),
		\\y_{N,\epsilon}(z)&=z,
	\end{align}
where $P^N_{1,\epsilon}\dots,P^N_{N,\epsilon}$ are small independent polynomial deformations of the polynomial $1+N^{-1}\psi(z)$ depending on a small parameter~$\epsilon$, such that all of these polynomials have the same degree as $\psi(z)$, and for $\epsilon\neq 0$ all zeros of $P^N_{1,\epsilon}\dots,P^N_{N,\epsilon}$ are simple and distinct (between all of these polynomials), and are close to the zeros of $1+N^{-1}\psi(z)$. The differential $dx_{N,\epsilon}$ is  meromorphic and the topological recursion is well defined. We obviously have
\begin{equation}
	\lim\limits_{N\to\infty}\lim\limits_{\epsilon\to 0}\psi_{N,\epsilon}(\ta)=\psi(\ta),\qquad
	\lim\limits_{N\to\infty}\lim\limits_{\epsilon\to 0}dx_{N,\epsilon}(z)=dx(z).
\end{equation}
According to~\cite[Proposition 3.1]{ABDKS4}, the TR differentials for this family of spectral curves are given by~\eqref{eq:y=z} with the function $\hat x_{N,\epsilon}$ taken in the form
\begin{equation}
	\hat x_{N,\epsilon}(z)=\frac{1}{\cS(\hbar\partial_z)}x_{N,\epsilon}(z).
\end{equation}
The formula is manifestly compatible with the limits $\epsilon\to 0$ and $N\to\infty$, where the first limit gives the $(\omega_N)^{(g)}_n$ differentials of~\cite{BKW} (it is a case of the standard compatibility of topological recursion with the limits discussed in detail in~\cite{limits}), and the subsequent $N\to\infty$ gives the formulas for the differentials computed with $\hat x$ given by \eqref{eq:OS-atlantes}, that is, the Orlov--Scherbin differentials enumerating Atlantes Hurwitz numbers (and not satisfying the topological recursion). 

\subsubsection{\texorpdfstring{$r$}{r}-spin Hurwitz numbers} Let us also illustrate how one can obtain the differentials enumerating $r$-spin Hurwitz numbers via a similar limit procedure. In this case the computation below is not reproducing the transalgebraic topological recursion of~\cite{BKW} as we lose the contribution of some of the critical points that they have for finite $N$. 

Let $\psi(\theta) = \theta^r$ and consider
\begin{align}
\psi_N(\ta)&=N\log(1+N^{-1}\psi(\ta)),
\\ x_N(z)&=\log z-\psi_N(z),
\\ y_N(z) & =z.
\end{align}
Then, $dx_N=\frac{dz}{z}-\frac{\psi'(z)\,dz}{1+N^{-1}\psi(z)}$ is rational and we have obviously
\begin{equation}
\lim\limits_{N\to\infty}\psi_N(\ta)=\psi(\ta),\qquad
\lim\limits_{N\to\infty}dx_N(z)=dx(z).
\end{equation}
%The possible multiple poles of $dx$ for this deformation of the spectral curve split into a number of simple poles. 
It is important to stress that the set of zeros of $dx_N$ does differ from the set of zeros of $dX_N$ given in~\eqref{eq:BKW-setup}, so the TR differentials for this family differ from the ones considered in the transalgebraic TR approach. 

According to~\cite[Proposition 3.1]{ABDKS4}, the TR differentials for this family of spectral curves are given again by~\eqref{eq:y=z} with the function $\hat x_N$ taken in the form
\begin{equation}
\hat x_N(z)=\frac{1}{\cS(\hbar\partial_z)}\log z-\frac{1}{\cS(N^{-1}\hbar\partial_z)}\psi_N(z).
\end{equation}
The differentials computed by this formula have a limit as $N\to\infty$, and the limiting differentials are given also by~\eqref{eq:y=z} with the function $\hat x$ given by~\eqref{eq:xhatspin}, that is, these are differentials for $r$-spin Hurwitz numbers.

\printbibliography

\end{document}